\documentclass[3p,sort&compress]{elsarticle}

\usepackage{amssymb, amsmath, amsthm}
\usepackage{xspace}
\usepackage{calc}
\usepackage{enumitem}
\usepackage{algorithm}
\usepackage{times}
\usepackage[T1]{fontenc}

\usepackage{amssymb}        
\usepackage{latexsym}
\usepackage{algpseudocode}
\usepackage{graphicx}
\usepackage{verbatim}

\usepackage{tikz}
\usepackage{subcaption}
\usetikzlibrary{decorations,math,calc,decorations.markings,decorations.pathmorphing,positioning,intersections,snakes}

\newtheorem{theorem}{Theorem}
\newtheorem{lemma}{Lemma}

\newtheorem{observation}{Observation}
\newtheorem{definition}{Definition}

\newtheorem{claim}{Claim}

\newcommand{\dowod}{\begin{proof}}
\newcommand{\koniec}{\end{proof}}

\newcommand{\lbmatching}{$(l,b)$-matching}

\newcommand{\kt}{\ensuremath{K_{t+1}}}
\newcommand{\ktt}{\ensuremath{K_{t,t}}}
\newcommand{\kpq}{\ensuremath{K^p_q}}
\newcommand{\kpt}{\ensuremath{K^p_2}}
\newcommand{\np}{\ensuremath{\mathcal{NP}}}
\newcommand{\mc}{\ensuremath{\bar{M}}}
\newcommand{\nc}{\ensuremath{\bar{N}}}
\newcommand{\gm}{\ensuremath{\hat{G}}}
\newcommand{\wm}{\ensuremath{\hat{w}}}
\newcommand{\mm}{\ensuremath{\hat{M}}}

\usepackage[textwidth=20mm,disable]{todonotes}
\newcommand{\mytodo}[2]{\todo[size=\tiny, color=#1!50!white]{#2}}

\newcommand{\mwcom}[1]{\mytodo{green}{#1}}
\newcommand{\revcom}[1]{\mytodo{red}{#1}}

\usepackage{hyperref}

\date{}

\begin{document}

\begin{frontmatter}

\title{Clique-free t-matchings in degree-bounded graphs\tnoteref{grant}}
\tnotetext[grant]{Partially supported by Polish National Science Center grant 2018/29/B/ST6/02633.}

\author[ii]{Katarzyna Paluch}
\ead{abraka@cs.uni.wroc.pl}

\author[ii]{Mateusz Wasylkiewicz}
\ead{mateusz.wasylkiewicz@cs.uni.wroc.pl}

\address[ii]{Institute of Computer Science,  University of Wroc{\l}aw, Wroc{\l}aw, Poland}


\begin{abstract}
We consider \mwcom{who has to be the corresponding author?} problems of finding a maximum size/weight $t$-matching without forbidden subgraphs in an undirected graph $G$ with the maximum degree bounded by $t+1$, where $t$ is an integer greater than $2$. Depending on the variant forbidden subgraphs denote certain subsets of $t$-regular complete partite subgraphs of $G$. A graph is complete partite if there exists a partition of its vertex set such that every pair of vertices from different sets is connected by an edge and  vertices from the same set form an independent set. \mwcom{maybe stress that the forbidden subgraphs does not have to be induced subgraphs} A clique \kt{} and a bipartite clique \ktt{} are examples of complete partite graphs. These problems are natural generalizations of the triangle-free and square-free $2$-matching problems in subcubic graphs. In the weighted setting we assume that the weights of edges of $G$ are \emph{vertex-induced} on every forbidden subgraph. We present simple and fast combinatorial algorithms for these problems. The presented algorithms are the first ones for the weighted versions, and for the unweighted ones,  are faster than those known previously. Our approach relies on the use of gadgets with so-called {\em half-edges}.  A {\em half-edge} of edge $e$ is, informally \mwcom{maybe add information that we calculate the complement of $t$-matching to the abstract} speaking, a half of $e$ containing exactly one of its \mwcom{$k$-restricted variant removed} endpoints. 
\end{abstract}


\end{frontmatter}

\pgfdeclaredecoration{complete sines}{initial} 
{
    \state{initial}[
        width=+0pt,
        next state=sine,
        persistent precomputation={\pgfmathsetmacro\matchinglength{
            \pgfdecoratedinputsegmentlength / int(\pgfdecoratedinputsegmentlength/\pgfdecorationsegmentlength)}
            \setlength{\pgfdecorationsegmentlength}{\matchinglength pt}
        }] {}
    \state{sine}[width=\pgfdecorationsegmentlength]{
        \pgfpathsine{\pgfpoint{0.25\pgfdecorationsegmentlength}{0.5\pgfdecorationsegmentamplitude}}
        \pgfpathcosine{\pgfpoint{0.25\pgfdecorationsegmentlength}{-0.5\pgfdecorationsegmentamplitude}}
        \pgfpathsine{\pgfpoint{0.25\pgfdecorationsegmentlength}{-0.5\pgfdecorationsegmentamplitude}}
        \pgfpathcosine{\pgfpoint{0.25\pgfdecorationsegmentlength}{0.5\pgfdecorationsegmentamplitude}}
}
    \state{final}{}
}

\tikzset{vertex/.style={minimum size=#1,circle,fill=black,draw,inner sep=0pt},
	decoration={markings,mark=at position .5 with {\arrow[black,thick]{stealth}}},
	vertex/.default=2.5mm,
	bigVertex/.style={vertex=4mm},
	e0/.style={line width=0.8pt},
	e1/.style={line width=2.4pt},
	em/.style={line width=0.8pt,decoration={
            complete sines,
            segment length=8,
            amplitude=5
        },
        decorate},
    em0/.style={line width=0.8pt,decoration={
            complete sines,
            segment length=12,
            amplitude=3
        },
        decorate},
	transformsTo/.pic=
	{
		\coordinate (-leftEnd) at (0,0);
		\coordinate (-rightEnd) at (5,0);
		\draw[thick] (0,0) -- (0,1) -- (3,1) -- (3,2) -- (5,0) -- (3,-2) -- (3,-1) -- (0,-1) -- (0,0);
	}}

\section{Introduction}
We consider several variants of the problem of finding a maximum size or weight $t$-matching without forbidden subgraphs.  

Given  a positive integer $t$, a subset~$M$ of edges of an undirected simple graph is a \emph{$t$-matching} if every vertex is incident to at most $t$ edges of~$M$. $t$-matchings belong to a wider class of $b$-matchings, where for every vertex $v$ in the set of vertices $V$ of the graph, we are given a natural number $b(v)$ and a subset of edges is a \emph{$b$-matching} if every vertex is incident to at most $b(v)$ of its  edges.  A $b$-matching of maximum size/weight can be found in polynomial time by a reduction to a classical matching.

In the problems that we study in this paper we are given an undirected graph $G=(V,E)$ in which each vertex has degree at most $t+1$ and the goal is to find a maximum size/weight $t$-matching  that does not contain certain complete $t$-regular partite graphs.
A graph $H=(V_H,E_H)$ is  {\bf \em complete partite} if there exists a partition $V_1$, $V_2$, \ldots, $V_p$ of $V_H$ of size $p\geq 2$ such that $E_H=\{(v,u):\exists i \neq j \;\; v\in V_i \land u\in V_j\}$. Each such $V_i$ is called a {\bf \em color class} of $H$. We denote the $i$-th color class of $H$ by $V_i(H)$. If all color classes of $H$ have the same size, we denote $H$ by $K^p_q$ where $q=|V_1|$. Notice that $K^p_q$ is $t$-regular for $t=(p-1)q$. Observe that a clique $K_a$ consisting of $a$ vertices and a bipartite clique $K_{a,b}$ with color classes consisting of $a$ and $b$ vertices are special cases
of complete partite graphs. 
We say that a $t$-matching of $G$ is {\bf \em restricted} if it does not contain an edge set of any \kt{} and \ktt{} of $G$ and that it is $K^p_q${\bf \em -free} for $t=(p-1)q$ if it does not contain an edge set of any subgraph of $G$ isomorphic to $t$-regular $K^p_q$. The restricted $t$-matching problem and the \kpq{}-free $t$-matching problem consist in finding a maximum size restricted/\kpq{}-free $t$-matching, respectively. Depending on the variant of the problem, we refer to  respective subgraphs  \kpq{}, \kt{} and \ktt{} as {\bf \em forbidden} subgraphs. We say that the restricted/\kpq{}-free $t$-matching problem is {\bf \em bounded} if every vertex of the input graph is assumed to have degree at most $t+1$.  Polynomial time algorithms for the bounded  restricted $t$-matching problem and the bounded \kpq{}-free $t$-matching  problem 
were presented by  B{\'e}rczi and V{\'e}gh~\cite{BercziVegh2010} and Kobayashi and Yin~\cite{KobayashiYin2012}, respectively. \\
\indent In the weighted versions  of these  problems,  each edge $e$ is associated with a  nonnegative weight $w(e)$ and we are interested in computing a restricted/\kpq{}-free  $t$-matching of maximum  weight.  For a subgraph $H$ of $G$, we say that the weight function is {\bf \em vertex-induced on $H$} if real (possibly negative) values called {\bf \em potentials} can be assigned to the vertices of $H$ in such a way that the weight of every edge of $H$ is equal to the sum of the potentials of its endpoints. Vornberger~\cite{Vornberger1980} showed that the weighted restricted $2$-matching problem is \np-hard. However, the weighted {\bf \em bounded}  restricted $2$-matching problem is solvable in polynomial time when the weights of the edges are vertex-induced on every subgraph isomorphic to a square, which was shown by Paluch and  Wasylkiewicz~\cite{PaluchWasylkiewicz2021a}. 

\subsection{Our results}
We present simple combinatorial algorithms for the weighted and unweighted versions of both the bounded  restricted $t$-matching problem and the bounded  \kpq{}-free $t$-matching problem, both for $t\geq 3$. In the weighted setting we assume that the weights of the edges are nonnegative and vertex-induced on every forbidden subgraph. In these algorithms, instead of calculating restricted or \kpq{}-free $t$-matching directly, we calculate its complement first. Such complement has fewer edges than sought-after restricted or \kpq{}-free $t$-matching which results in faster running time of the algorithm. To accomplish this, we augment some forbidden subgraphs of the input graph with gadgets containing so-called {\em half-edges} and define a function $b$ on the set of vertices in such a way that, any $b$-matching in the thus constructed graph $G'$ yields a complement of restricted or \kpq{}-free $t$-matching. 
 Half-edges have already been  introduced in \cite{PaluchEtAl2012} and used in several subsequent papers.  
The presented algorithms  are the first ones for the weighted versions of these problems. The running time of our algorithms is $\mathcal{O}(\min\{nm\log{n},n^3\})$.  Moreover, for the unweighted version, our algorithms are faster than those known previously. The algorithm for the unweighted bounded  restricted $t$-matching problem given by B{\'e}rczi and V{\'e}gh~\cite{BercziVegh2010} runs in $\mathcal{O}(t^3n+nm\log{m})$ time whereas the algorithm for the unweighted bounded  \kpq{}-free $t$-matching problem given by Kobayashi and Yin~\cite{KobayashiYin2012} runs in $\mathcal{O}(pq^3n+nm\log{m})$ time.
For the unweighted versions of both problems, the running times of our algorithms are  $\mathcal{O}(\sqrt{n}m)$. Previous approaches  to these problems relied on the shrinking  of parts of or subsets of forbidden subgraphs and   because of this   do not lend themselves to the weighted  setting. \\
\indent The algorithms presented  in this paper extend and generalize those from \cite{PaluchWasylkiewicz2021a}. The generalizations  require some new ideas as well as making observations regarding the structure.  It turns out that the assumption $t \geq 3$ allows us to simplify some arguments. The most involved is the case of $q=2$ in the bounded  \kpq{}-free $t$-matching problem.  Notice that this case was considered separately also in the previous algorithm for this problem given by Kobayashi and \mwcom{$k$-restricted variant removed} Yin~\cite{KobayashiYin2012}.

\subsection{Related work}
Restricted  $2$-matchings and restricted $t$-matchings are classical problems of combinatorial optimization. Notice that a $2$-matching is a set of vertex-disjoint cycles and paths whereas a restricted $2$-matching is a $2$-matching whose every cycle has length at least five since $K_3$ is a triangle, i.e. cycle of length three, and $K_{2,2}$ is a square, i.e. cycle of length four. Therefore, the  restricted $2$-matching problem can be used to approximate some variants of the travelling salesman problem (see~\cite{FisherEtAl1979} for details). Hartvigsen gave a complicated algorithm for the problem of finding a  maximum size triangle-free $2$-matching. Papadimitriou \cite{CornuejolsPulleyblank1980} showed that it is \np-hard to find a  maximum size $2$-matching without any cycle of length at most five. Recently, Kobayashi~\cite{Kobayashi2020} gave a polynomial algorithm for finding a maximum weight $2$-matching that does not contain any triangle from a given set of forbidden edge-disjoint triangles. Polynomial algorithms for the  square-free and/or \mwcom{maybe define "square-free" and "triangle-free"?} triangle-free $2$-matchings in subcubic graphs were presented in \cite{HartvigsenLi2011,HartvigsenLi2012,BercziKobayashi2012,BercziVegh2010,Kobayashi2010,PaluchWasylkiewicz2021a,KobayashiYin2012}. The weighted version of the restricted $2$-matching problem is \np-hard for general weights, even in graphs which are both subcubic and bipartite (see~\cite{BercziKobayashi2012}), therefore an assumption is made in all of these algorithms that the weight function is vertex-induced on every forbidden square. Such weight functions can be seen as a generalization of weight functions which assign values to the vertices instead of the edges. Regarding the square-free $2$-matching problem in general graphs, Nam~\cite{Nam1994} constructed a complex algorithm for it for graphs, in which all squares are vertex-disjoint. Finding a polynomial algorithm for the  square-free $2$-matching problem in general graphs is an open problem.

A natural generalization of the restricted $2$-matching problem is the restricted $t$-matching problem. Since the weighted version of the restricted $2$-matching problem is \np-hard already, in the weighted version of the restricted $t$-matching problem  an assumption is  typically made that the weight function is vertex-induced on every forbidden subgraph. Polynomial algorithms for restricted $t$-matchings in bipartite graphs were presented in \cite{Hartvigsen2006,Pap2007,Makai2007,Takazawa2009,PaluchWasylkiewicz2021b}. In all previous algorithms for the weighted version of the restricted $t$-matching problem  in bipartite graphs, an assumption is made that the weight function is vertex-induced on every forbidden \ktt{}.  
It is an interesting question to generalize the known results for the square-free and/or triangle-free $2$-matching problem in subcubic graphs to restricted $t$-matchings. Since all of these algorithms used the fact that every node of a forbidden subgraph is incident to at most one edge which does not belong to this forbidden subgraph, it is natural to consider restricted $t$-matching problem in  graphs of the maximum degree at most $t+1$, i.e. the bounded restricted $t$-matching problem.  As stated before,  the unweighted version of this problem was considered by B{\'e}rczi and V{\'e}gh~\cite{BercziVegh2010}.

It is worth mentioning that the bounded restricted $t$-matching problem has  some connections to the (vertex-)connectivity augmentation problem in which we want to make a  given undirected graph $G$  $k$-connected by adding a minimum number of different  new edges. Let $n$ be the number of the vertices of $G$.  As pointed out by B{\'e}rczi and V{\'e}gh~\cite{BercziVegh2010}, this problem is equivalent to finding a maximum size $t$-matching which contains  no complete bipartite subgraphs $K_{a,b}$ consisting of $a+b=t+2$ vertices in the complement of $G$ for $t=n-k-1$. A special case of the connectivity augmentation problem is increasing connectivity by one problem when given graph $G$ is already $(k-1)$-connected. There exists a polynomial algorithm for this case given by V\'{e}gh \cite{Vegh2011} which works in $\mathcal{O}(kn^7)$ time. Notice that the maximum degree of the complement of a $(k-1)$-connected graph is at most $t+1$ for $t=n-k-1$, therefore the increasing connectivity by one problem can be reduced to a problem of finding a maximum size $t$-matching which contains  no complete bipartite subgraphs of size $t+2$ in the graphs of the maximum degree at most $t+1$. Weighted versions of these problems are also of interest since, in the same paper, V\'{e}gh considered a generalization of the increasing connectivity by one problem when the edges of the complement of $G$ are given weights and we want to minimize the total sum of added edges. This version is \np-hard for general weights, however V\'{e}gh presented a polynomial algorithm for the weighted version if the weights are vertex-induced on the whole complement of $G$.

It turns out that the polynomial solvability of the restricted $t$-matching problem has connections to jump systems. It was conjectured by Cunningham~\cite{Cunningham2003} that for every natural number $k$, the degree sequences of $2$-matchings with no cycles of length at most $k$ of any graph form a jump system if and only if there exists a polynomial algorithm for finding such $2$-matchings of maximum size. The only open case of this conjecture is $k=4$. It was shown by Kobayashi et al.~\cite{KobayashiEtAl2012} that the degree sequences of restricted $2$-matchings of any graph form a jump system which suggests that there should exist a polynomial algorithm for the restricted $2$-matching problem in general graphs. A generalization of Cunningham's conjecture to the $t$-matchings seems to be also true.  In the same paper, Kobayashi et al. proved that the degree sequences of  restricted $t$-matchings in any graph form a jump system. In fact, they proved an even stronger result. They showed that the degree sequences of $t$-matchings with no subgraphs from a given family $\mathcal{H}$ of $t$-regular connected graphs form a jump system in any graph if and only if every member of $\mathcal{H}$ is a complete partite graph. Moreover, it was shown by Kobayashi and Yin~\cite{KobayashiYin2012} that this problem is \np-hard, even in the graphs of  maximum degree at most $t+1$, if $\mathcal{H}$ is a singleton of any $t$-regular connected graph which is not complete partite. Notice that a $t$-regular complete partite graph is exactly \kpq{} for $t=(p-1)q$, hence it is natural to consider \kpq{}-free $t$-matchings for such values of $t$.

The weighted versions of restricted $t$-matchings were studied in the context of the $M$-concave functions on constant-parity jump systems, which can be seen as discrete counterparts of the submodular functions. Kobayashi et al.~\cite{KobayashiEtAl2012} proved that the degree sequences of maximum weight restricted $t$-matchings in a bipartite graph form an $M$-concave function on the constant-parity jump system if and only if the weight function is vertex-induced on every forbidden subgraph. The algorithm of B{\'e}rczi and Kobayashi~\cite{BercziKobayashi2012} for maximum weight square-free $2$-matchings in subcubic graphs is based on the fact that the degree sequences of these $2$-matchings form an $M$-concave function in any subcubic graph if the weight function is vertex-induced on every square. On the other hand, the degree sequences of maximum weight triangle-free $2$-matchings form an $M$-concave function in any graph, which was shown by Kobayashi~\cite{Kobayashi2014}. However, it is unknown if there exists a polynomial algorithm for finding maximum weight triangle-free $2$-matchings in general  graphs.

\section{Preliminaries}
Let $G=(V,E)$ be an undirected  graph with vertex set~$V$ and edge set~$E$. We denote the number of vertices of $G$ by $n$ and the number of edges of $G$ by $m$. We assume that all graphs are {\bf \em simple}, i.e., they contain neither loops nor parallel edges. We denote an edge connecting vertices $v$ and $u$ by $(v,u)$. For a subgraph $H$ of $G$, we denote the vertex set of $H$ by $V(H)$ and the edge set by $E(H)$. Given a weight function $w:E\to\mathbb{R}$, the {\bf \em weight of $H$}, denoted by $w(H)$, is defined as the sum of weights of edges of $H$. For an edge set $F\subseteq E$ and $v\in V$, we denote by $\deg_F(v)$ the number of edges of~$F$ incident to $v$. With a slight abuse of notation, for a graph $H$, we sometimes write $\deg_H(v)$ instead of $\deg_{E(H)}(v)$. For any natural number $k$, we say that $G$ is $k$-regular if every vertex of $G$ has degree exactly $k$.
For a vertex set $A\subseteq V$, we denote by $G[A]$ a subgraph of $G$ induced \revcom{The grammar problems are mostly about the wrong use or false omission of definite and/or indefinite articles} by $A$.  Given two graphs $H_1=(V_1,E_1)$ and $H_2=(V_2,E_2)$, we denote by $H_1 \cap H_2$ a graph $(V_1 \cap V_2, E_1 \cap E_2)$. With a slight abuse of notation, we denote that $H_1$ and $H_2$ are isomorphic by $H_1=H_2$.

An instance of each of the two problems that we consider in the paper consists of an undirected graph~$G=(V,E)$ whose every vertex has degree at least one  and at most $t+1$, and a weight function $w:E\to\mathbb{R}_{\geq 0}$. In the bounded  $K^p_q$-free $t$-matching problem we are also given natural numbers $p\geq 2$, $q\geq 1$ such that $t=(p-1)q$. In the bounded  restricted (resp. $K^p_q$-free) $t$-matching problem we assume that $w$ is vertex-induced on every \kt{} and \ktt{} (resp. on every $K^p_q$) of $G$ and the goal is to find a maximum weight restricted (resp. $K^p_q$-free) $t$-matching of~$G$.

For  a weight function $w:E\to\mathbb{R}$ which is vertex-induced on some subgraph $H$ of $G$, we say that a function $r_H: V(H)\to\mathbb{R}$ which assigns a potential to every vertex of $H$ is a {\bf \em potential function} of $H$.

We say that an edge set $M\subseteq E$ is a {\bf \em co-$t$-matching} of $G$ if $E\setminus M$ is a $t$-matching of $G$. Additionally, we say that a co-$t$-matching $M\subseteq E$ is {\bf \em covering} (resp. {\bf \em \kpq{}-covering}) if $E\setminus M$ is a restricted (resp. \kpq{}-free) $t$-matching of $G$. We say that an edge set $M\subseteq E$ {\bf \em covers} a forbidden subgraph $H$ of $G$, or that $H$ {\bf \em is covered by} $M$, if $M$ contains at least one edge of $H$. Notice that if the degree of every vertex of $G$ is at most $t+1$, then an edge set $M\subseteq E$ is a co-$t$-matching of $G$ if every vertex of $G$ of degree $t+1$ is incident to at least one edge of $M$. Moreover, $M$ is a \kpq{}-covering (resp. covering) co-$t$-matching if, in addiction, $M$ covers every \kpq{} (resp. every \kt{} and \ktt{}) of $G$. A (\kpq{}-)covering co-$t$-matching $M$ of $G$ is said to be {\bf \em minimum weight} if there is no (\kpq{}-)covering co-$t$-matching $N$ of~$G$ of weight smaller than $w(M)$.

We will need to compute a $b$-matching of a graph $G$ where we are given vectors $l,b\in\mathbb{N}^V$ and a weight function~$w:E\to\mathbb{R}$. (We allow negative weights here.) For a vertex $v\in V$, $[l(v),b(v)]$ is said to be a {\bf \em capacity interval} of $v$. An edge set~$M\subseteq E$ is said to be an \lbmatching{} if $l(v) \leq \deg_M(v) \leq b(v)$ for every $v\in V$. Given an \lbmatching{} $M$ and an edge $e=(u,v) \in M$, we say that $u$ is {\bf \em matched to} $v$ in $M$. Moreover, we say that a vertex $v$ of $G$ is {\bf \em unmatched} in $M$ if $v$ is not incident to any edge of $M$. A {\bf \em matching} and a {\bf \em perfect matching} of $G$ is any \lbmatching{} of $G$ where the capacity interval of every vertex $v$ of $G$ is equal to, respectively, $[0,1]$ and $[1,1]$. An \lbmatching{} $M$ is said to be a {\bf \em maximum weight} \lbmatching{} if there is no \lbmatching{}~$M'$ of~$G$ of weight greater than $w(M)$. A maximum weight \lbmatching{} can be computed efficiently.

\begin{theorem}[\cite{Gabow1983}]\label{thm:gabow}
There is an algorithm that, given a multigraph~$G=(V,E)$ (i.e. $G$ may contain parallel edges), a weight function~$w:E\to\mathbb{R}$ and vectors $l,b\in\mathbb{N}^V$, finds a maximum weight \lbmatching{} of $G$ in  time $O\left((\sum_{v\in V}b(v))\min\{|E|\log{|V|},|V|^2\}\right)$, assuming that $G$ admits any \lbmatching{}.
\end{theorem}

We define a {\bf \em minimum weight} \lbmatching{} of $G$ analogously. Notice that calculating a minimum weight \lbmatching{} of $G$ can be reduced to calculating a maximum weight \lbmatching{} of $G$ by negating the weight of every edge of $G$.



\section{Outline of the Algorithm}\label{sec:outline}
We give an outline of our algorithms for the weighted versions of the bounded  restricted and the bounded  \kpq{}-free $t$-matching problem assuming that the weight function  is vertex-induced on every forbidden subgraph.  Recall that we assume that $t\geq 3$. In the \kpq{}-free $t$-matching problem we additionally assume that $t=(p-1)q$.


\begin{algorithm}[H]
\begin{algorithmic}[1]
	\State \label{itm:step1}Construct an auxiliary multigraph~$G'=(V',E')$ consisting of $O(n)$ vertices and $O(m)$ edges by augmenting some forbidden subgraphs of $G$ with gadgets containing 
	half-edges. (Both gadgets and half-edges are defined later.)
	\State \label{itm:step2}Define a weight function~$w':E'\to\mathbb{R}$ and vectors $l,b\in\mathbb{N}^{V'}$.
	\State \label{itm:step3}Compute a minimum weight \lbmatching{} $M'$ of $G'$.
	\State \label{itm:step4}Construct a co-$t$-matching~$\mc$ of~$G$ by replacing all half-edges of $M'$ with some edges of~$G$ in such a way that $w(\mc) \leq w'(M')$.
	\State \label{itm:step5}Add some edges of the remaining forbidden subgraphs to $\mc$ by replacing some of its edges with other ones without increasing the weight of $\mc$.
	\State Return $E\setminus\mc$.
\end{algorithmic}
\caption{Computing a maximum weight restricted or \kpq{}-free $t$-matching of graph~$G$ of maximum degree at most $t+1$ given a weight function \mbox{$w:E\to\mathbb{R}_{\geq 0}$}.}
\label{alg:main}
\end{algorithm}


The precise \revcom{Section 3 is difficult to read without the details of the algorithm, which is described later in the paper.
Especially, it is impossible to follow the proof of Claim on the complexity of the algorithm at this point.
I would suggest to reconsider the organization of the paper.} construction of $G'$, $w'$, $l$ and $b$ for the bounded  restricted (resp. \kpq{}-free) $t$-matching problem is given in Section~\ref{sec:cliques} (resp. Section~\ref{sec:partite}).

For the bounded  restricted $t$-matching problem, an implementation of Step~\ref{itm:step4} is given in Theorem~\ref{thm:cliques-matching-to-opt} whereas an implementation of Step~\ref{itm:step5} consists of subsequent applications of Lemma~\ref{lem:cliques-removing}. An implementation of these steps for the bounded  \kpq{}-free $t$-matching problem is given in Theorem~\ref{thm:partite}.

\begin{claim}\label{fact:running-time}
 Algorithm~\ref{alg:main}  runs in time $\mathcal{O}(\min\{nm\log{n},n^3\})$ in the weighted variant and $\mathcal{O}(\sqrt{n}m)$ in the unweighted.
\end{claim}

We prove Claim~\ref{fact:running-time} in Section~\ref{sec:running}.

\tikzset{
	kt/.pic=
	{
		\tikzmath{\a=3;}
		\foreach \i in {0,...,5}
		{
			\node (v\i)	at (\i*60:\a)	[bigVertex]{};
			\draw[e0] (v\i) -- (\i*60:1.5*\a);
		}
		\node also [label={[font=\Huge]left:$v_1$}] (v2);
		\node also [label={[font=\Huge]right:$v_2$}] (v1);
		\foreach \i in {0,...,5}
			\foreach \j in {0,...,5}
			{
				\ifnum \i = \j {}
				\else
					\draw[e0] (v\i) -- (v\j);
				\fi
			}
	},
	ktGadget/.pic=
	{
		\tikzmath{\a=2;}
		\foreach \i in {0,...,5}
			\node (v\i) at (\i*60:3*\a)	[bigVertex]{};
		\node also [label={[font=\Huge]left:$v_1$}] (v2);
		\node also [label={[font=\Huge]right:$v_2$}] (v1);
		\node (z) at (0,0) [vertex,label={[font=\Huge,xshift=5]below right:$u_H$}]{};
		\foreach \i in {0,...,5}
			\draw[e0] (z) -- (v\i);
	},
	ktt/.pic=
	{
		\tikzmath{\a=3;}
		\begin{scope}[shift={(0,-\a)}]
		\foreach \k in {0,1}
			\foreach \i in {0,...,2}
			{
				\tikzmath{int \s; \s = 2*\k-1;}
				\node (v\k\i) at (\s*\a,\i*\a)	[bigVertex]{};
				\draw[e0] (v\k\i) -- (\s*\a+\s*\a/2,\i*\a);
			}
		\node also [label={[font=\Huge]above:$a_1$}] (v02);
		\node also [label={[font=\Huge]above:$b_1$}] (v12);
		\foreach \i in {0,...,2}
			\foreach \j in {0,...,2}
				\draw[e0] (v0\i) -- (v1\j);
		\end{scope}
	},
	kttGadget/.pic=
	{
		\tikzmath{\a=6; \b=5;}
		\begin{scope}[shift={(0,-\b)}]
		\foreach \k in {0,1}
			\foreach \i in {0,...,2}
			{
				\tikzmath{int \s; \s = 2*\k-1;}
				\node (v\k\i) at (\s*\a,\i*\b)	[bigVertex]{};
			}
		\node also [label={[font=\Huge]left:$a_1$}] (v02);
		\node also [label={[font=\Huge]right:$b_1$}] (v12);
		\foreach \k in {0,1}
		{
			\tikzmath{int \t; \t=1-\k;}
			\node (u\k) at ($(v\k1)!0.3!(v\t1)$) [vertex]{};
			\foreach \i in {0,...,2}
				\draw[e0] (v\k\i) -- (u\k);
		}
		\node also [label={[font=\Huge]right:$u_H^1$}] (u0);
		\node also [label={[font=\Huge]left:$u_H^2$}] (u1);
		\end{scope}
	}
}

\tikzset{
	 en/.style={line width=0.8pt,decoration={
            complete sines,
            segment length=5,
            amplitude=2
        },
    decorate},
	pics/flipping/.style n args={9}{code={
		\tikzmath{\a=1;}
		\coordinate (-center) at (0,0);
		\begin{scope}[font=\Huge,shift={(0,-\a)}]
		\foreach \i in {0,1}
		{
			\node (v0) at (-4*\a,0) [vertex,label=below:$a$]{};
			\node (v1) at (4*\a,0) [vertex,label=below:$b$]{};
			
			\draw[en] (v0) -- (-5*\a,-0.8*\a);
			\draw[e0] (v1) -- (5*\a,-0.8*\a);
			
			\node (u0) at (-\a,3*\a) [vertex,label=above:$u_H^1$]{};
			\node (u1) at (\a,3*\a) [vertex,label=above:$u_H^2$]{};
			
			\node (h00) at (-2*\a,0) [vertex,label=below:$v^a_b$]{};
			\node (h10) at (2*\a,0) [vertex,label=below:$v_b^a$]{};
			
			\foreach \j in {1,...,3}
				\node (h0\j) at ($(-4*\a,0)+(30*\j:2*\a)$) [vertex]{};
			\foreach \j in {1,...,3}
				\node (h1\j) at ($(4*\a,0)-(-30*\j:2*\a)$) [vertex]{};
			
			\draw[#1] (h00) -- (h10);
			
			\draw[#2] (v0) -- (h00);
			\draw[#4] (v0) -- (h01);
			\draw[en] (v0) -- (h02);
			\draw[en] (v0) -- (h03);
			
			\draw[#6] (u0) -- (h00);
			\draw[#8] (u0) -- (h01);
			\draw[e0] (u0) -- (h02);
			\draw[e0] (u0) -- (h03);
			
			\draw[#3] (v1) -- (h10);
			\draw[en] (v1) -- (h11);
			\draw[#5] (v1) -- (h12);
			\draw[e0] (v1) -- (h13);
			
			\draw[#7] (u1) -- (h10);
			\draw[e0] (u1) -- (h11);
			\draw[#9] (u1) -- (h12);
			\draw[e0] (u1) -- (h13);
			
			\draw[en] (h13) -- +(0,2*\a);
		}
		\end{scope}
	}},
	pics/flip1/.style n args={4}{code=
	{
		\tikzmath{\a=1;\b=2;}
		\begin{scope}[font=\Huge,shift={(0,-1.5*\b)}]
		\node (a) at (0,3*\b)		[bigVertex,label=above:$a$]{};
		\node (b) at (3*\a,0)		[bigVertex,label=below:$b$]{};
		\node (c) at (0,0)			[bigVertex,label=below:$a'$]{};
		\node (vca) at (0,\b)		[vertex]{};
		\node (vac) at (0,2*\b)		[vertex]{};
		\node (vab) at (\a,2*\b)	[vertex]{};
		\node (vba) at (2*\a,\b)	[vertex]{};
		
		\draw[en] (c) -- (vca);
		\draw[e0] (vca) -- (vac);
		\draw[#1] (vac) -- (a);
		\draw[#2] (a) -- (vab);
		\draw[#3] (vab) -- (vba);
		\draw[#4] (vba) -- (b);
		
		\end{scope}
	}},
	flip1a/.pic=
	{
		\tikzmath{\a=1;}
		\coordinate (-center) at (0,0);
		\pic (f) {flip1={e0}{en}{e0}{e0}};
		\draw[en] (fvac) -- +(-1.5*\a,0);
		\draw[en] (fvba) -- +(1.5*\a,0);
	},
	flip1b/.pic=
	{
		\coordinate (-center) at (0,0);
		\pic (f) {flip1={en}{e0}{en}{e0}};
	},
	pics/flip2/.style n args={5}{code=
	{
		\tikzmath{\a=1;\b=2;}
		\begin{scope}[font=\Huge,shift={(0,-1.5*\b)}]
		\node (a) at (0,3*\b)		[bigVertex,label=above:$a$]{};
		\node (b) at (3*\a,0)		[bigVertex,label=below:$b$]{};
		\node (c) at (0,0)			[bigVertex,label=below:$a'$]{};
		\node (d) at (3*\a,3*\b)	[bigVertex,label=above:$b'$]{};
		\node (vca) at (0,\b)		[vertex]{};
		\node (vac) at (0,2*\b)		[vertex]{};
		\node (vab) at (\a,2*\b)	[vertex]{};
		\node (vba) at (2*\a,\b)	[vertex]{};
		\node (vbd) at (3*\a,\b)	[vertex]{};
		\node (vdb) at (3*\a,2*\b)	[vertex]{};
		
		\draw[en] (c) -- (vca);
		\draw[e0] (vca) -- (vac);
		\draw[#1] (vac) -- (a);
		\draw[#2] (a) -- (vab);
		\draw[#3] (vab) -- (vba);
		\draw[#4] (vba) -- (b);
		\draw[#5] (b) -- (vbd);
		\draw[e0] (vbd) -- (vdb);
		\draw[en] (vdb) -- (d);
		\end{scope}
	}},
	flip2a/.pic=
	{
		\tikzmath{\a=1;}
		\coordinate (-center) at (0,0);
		\pic (f) {flip2={e0}{en}{e0}{en}{e0}};
		\draw[en] (fvac) -- +(-1.5*\a,0);
		\draw[en] (fvbd) -- +(1.5*\a,0);
	},
	flip2b/.pic=
	{
		\coordinate (-center) at (0,0);
		\pic (f) {flip2={en}{e0}{en}{e0}{en}};
	}
}


\section{The bounded restricted $t$-matching  problem} \label{sec:cliques}


In this section we address the weighted bounded restricted $t$-matching  problem. Recall that we assume that $t\geq 3$ and the nonnegative weight function $w$ is vertex-induced on every \kt{} and \ktt{} of $G$. \\
\indent We build a graph $G'=(V', E')$ together with a weight function \mbox{$w':E'\to\mathbb{R}$}, in which subgraphs, called  {\bf \em gadgets}, are added to some forbidden subgraphs of $G$. The precise construction of gadgets is the following. \\
\indent Let $H$ be any $K_{t+1}$ of $G$. We define a {\bf \em gadget for $H$}. Let $v_1$, $v_2$, \ldots, $v_{t+1}$ be vertices of $H$. We introduce one new vertex $u_H$ called a {\bf \em subdivision vertex}. For each vertex $v_i$ of $H$ we add edge $(v_i,u_H)$ called a {\bf \em half-edge} of $H$. We set the capacity interval of $u_H$ to $[2,2]$.


\begin{figure}[htpb]
\centering
\begin{tikzpicture}[scale=0.35,transform shape]
	\pic (arrow) [scale=0.5]{transformsTo};
	\pic (H) [left=7cm of arrow-leftEnd] {kt};
	\pic (g) [right=9cm of arrow-rightEnd] {ktGadget};
\end{tikzpicture}
\caption{A gadget for  \kt{} for $t=5$.}
\end{figure}
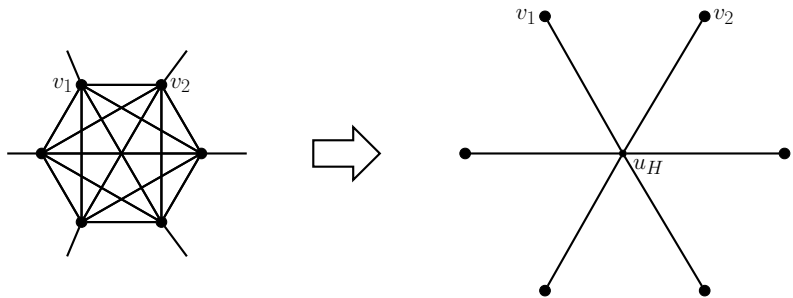



Let $H$ be any \ktt{} of $G$. We define a {\bf \em gadget for $H$}. Let $a_1$, $a_2$, \ldots, $a_t\in V_1(H)$  and $b_1$, $b_2$, \ldots, $b_t \in V_2(H)$. We introduce two new subdivision vertices $u_H^1$ and $u_H^2$. We connect $u_H^1$ with all vertices of $V_1(H)$ by half-edges. Symmetrically, we connect $u_H^2$ with all vertices of $V_2(H)$ by half-edges. We set the capacity intervals of $u_H^1$ and $u_H^2$ to $[1,1]$.

The main ideas behind  these gadgets  are the following. An \lbmatching{} $M'$ of $G'$ is to represent roughly a covering co-$t$-matching $\mc$ of $G$. We want to ensure that at least one edge of $H$ belongs to $\mc$. Notice that $M'$ contains exactly two half-edges of $H$ incident to some vertices $z_1$ and $z_2$ of $H$. We add an edge $(z_1,z_2)$ to $\mc$. Moreover, we add to $\mc$ all edges of $H$ which belong to $M'$. In this way, $H$ is guaranteed to be covered by $\mc$.

Regarding weights  $w'$ of the edges in the gadgets, we assign them as follows. 
Let $r_H$ be any potential function of $H$.
Every half-edge of $H$ incident to vertex $v_i$ gets weight $r_H(v_i)$.

\begin{figure}[htpb]
\centering
\begin{tikzpicture}[scale=0.35,transform shape]
	\pic (arrow) [scale=0.5]{transformsTo};
	\pic (H) [left=7cm of arrow-leftEnd] {ktt};
	\pic (g) [right=11cm of arrow-rightEnd] {kttGadget};
\end{tikzpicture}
\caption{A gadget for  \ktt{} for $t=3$.}
\end{figure}
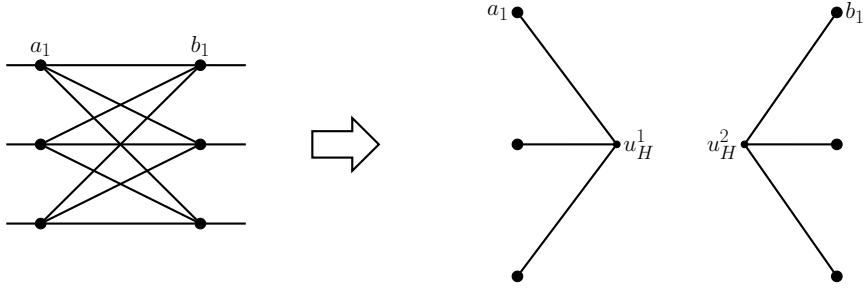


\subsection{Disjoint forbidden subgraphs}\label{sec:disjoint-forbidden-cliques}

In this subsection we consider a case when all forbidden subgraphs of $G$ are pairwise vertex-disjoint. We show how to drop this assumption later, in Subsection~\ref{sec:non-disjoint-forbidden-cliques}.

The precise construction of $G'$, $w'$, $l$ and $b$ is the following. We start off with $G'=G$.  We add a gadget to every forbidden subgraph of $G$. This operation is well-defined because these forbidden subgraphs are pairwise vertex-disjoint. We set the capacity interval of every vertex $v$ of degree $t+1$ of $G$ to $[1,t+1]$. We set the capacity interval of every other vertex $v$ of $G$ to $[0,\deg_G(v)]$. We set the weight of every edge $e$ of $G'$ which belongs to $G$ to $w(e)$.
In theorems below we show the correspondence between covering co-$t$-matchings of $G$ and {\lbmatching{}}s of  $G'$.

\begin{theorem}\label{thm:cliques-opt-to-matching}
Let $\mc$ be any covering co-$t$-matching of $G$. Then we can find an \lbmatching{} $M'$ of $G'$ such that $w'(M') = w(\mc)$.
\end{theorem}
\dowod
We initialize $M'$ as the empty set. We add every edge of $\mc$ that does not belong to any forbidden subgraph of $G$ to $M'$.

Consider any $H=K_{t+1}$ of $G$. Since $\mc$ is covering, there exists an edge of $H$ which belongs to $\mc$. If more than one edge of $H$ belongs to $\mc$, we choose one of them. Suppose that we chose $(z_1,z_2)\in\mc$. Then we add edges $(z_1,u_H)$ and $(z_2,u_H)$ to $M'$. We add every other edge of $H$ which belongs to $\mc$ to $M'$.

Consider any $H=K_{t,t}$  of $G$.  Since $\mc$ is covering, there exists an edge $(a,b)$ of $H$ which belongs to $\mc$. Suppose that $a\in V_1(H)$ and $b \in V_2(H)$.  We add edges $(a,u_H^1)$ and $(b,u_H^2)$ to $M'$. We add every other edge of $H$ which belongs to $\mc$ to $M'$.

Since the weight of any replaced edge in $G$ is equal to the sum of the weights of its corresponding half-edges in $G'$, we get that $w'(M') = w(\mc)$.
\koniec

\begin{theorem}\label{thm:cliques-matching-to-opt}
Let $M'$ be any \lbmatching{} of $G'$. Then we can find a covering co-$t$-matching~$\mc$ of $G$ such that $w(\mc) \leq w'(M')$.
\end{theorem}
\dowod
We initialize $\mc$ as the empty set. We add every edge of $M'$ which belongs to $G$ to $\mc$. For every forbidden subgraph of $G$ we add some of its edges to $\mc$.

Let $H$ be any \kt{} of $G$. Notice that exactly two vertices of $H$, say $v_1$ and $v_2$, are matched to $u_H$ in $M'$. We add $(v_1,v_2)$ to $\mc$ if it does not belong to $\mc$ already.

For any $H=K_{t,t}$ of $G$ we proceed analogously, because in the gadget for $H$ the subdivision vertices $u_H^1$ and $u_H^2$ are matched to two vertices $a\in V_1(H)$ and $b\in V_2(H)$. \koniec


\subsection{Non-disjoint forbidden subgraphs}\label{sec:non-disjoint-forbidden-cliques}
In the following lemmas we examine the ways in which forbidden subgraphs of $G$ may  overlap. 

\begin{lemma}\label{lem:cliques-kt-description}
Let $H_1$ and $H_2$ be any two different \kt{}'s of $G$ with a common vertex. Then $H_1 \cap H_2 =K_t$.
\end{lemma}
\dowod
Let $v$ be a common vertex of  $H_1$ and $H_2$. It has at most $t+1$ neighbours and exactly $t$ of them belong to $H_1$ and $t$ to $H_2$. Therefore, at least $t-1$ neighbours of $v$ belong to both $H_1$ and $H_2$, so $H_1$ and $H_2$ have at least $t$ common vertices. Since $H_1$ and $H_2$ are different, they cannot have $t+1$ common vertices.
\koniec


\begin{lemma}\label{lem:cliques-ktt-description}
Let $H_1$ and $H_2$ be any two different \ktt{}'s of $G$ with a common vertex. Then $H_1 \cap H_2$ contains $K_{t-1, t-1}$.
\end{lemma}
\dowod
Let $v \in V_1(H_1)$ be a common vertex of  $H_1$ and $H_2$. W.l.o.g. we can assume that $v \in V_1(H_2)$ (otherwise we can rename the color classes of $H_2$). The vertex $v$ is incident to exactly $t$ edges  of $H_1$ and $t$ edges of $H_2$, hence at least $t-1$ of these edges belong to $H_1 \cap H_2$. Let $U$ denote their endpoints. Then $U \setminus \{v\} \subseteq V_2(H_1) \cap V_2(H_2)$ and $ |U \setminus \{v\}| \geq t-1\geq 1$. This implies that there exists a vertex $v' \in V_2(H_1) \cap V_2(H_2)$ and thus, by repeating the same argument, there are at least $t-1$ vertices in $V_1(H_1) \cap V_1(H_2)$.
\koniec


\begin{lemma}\label{lem:cliques-kt-ktt-description}
Let $H_1$ and $H_2$ be any, respectively, \kt{} and \ktt{} of $G$ with a  common vertex. Then $t=3$ and $H_1 \cap H_2=K_{2,2}$.
\end{lemma}
\dowod
Let $v_1$ be a common vertex of $H_1$ and $H_2$. W.l.o.g we can assume that $v_1\in V_1(H_2)$. Repeating the same argument as in the proofs of Lemma~\ref{lem:cliques-kt-description} and Lemma~\ref{lem:cliques-ktt-description}, we get that at least $t-1\geq 1$ vertices of $V_2(H_2)$ belong to $H_1$. Similarly, at least $t-1$ vertices of $V_1(H_2)$ belong to $H_1$, so $H_1$ has at least $2t-2$ vertices. Therefore, $2t-2\leq t+1$, so $t=3$, and $H_1$ has exactly two common vertices with each color class of $H_2$.
\koniec


We define the following family of {\em \bf problematic}  subgraphs of $G$.


\begin{definition}\label{def:cliques-unproblematic-kt}
Let $H$ be any forbidden \kt{} of $G$. We say that $H$ is {\bf \em unproblematic} if $H$ has a common vertex  with another forbidden  subgraph $H'$ of $G$ such that:
\begin{enumerate}
\item\label{itm:cliques-unproblematic-kt-weight}  $H'=K_{t+1}$ and $w(H)\leq w(H')$, or
\item\label{itm:cliques-unproblematic-kt-3} $H'=K_{t,t}$.
\end{enumerate}
Otherwise, we say that $H$ is {\bf \em problematic}.
\end{definition}

\begin{definition}\label{def:cliques-unproblematic-ktt}
Let $H$ be any forbidden \ktt{} of $G$. We say that $H$ is {\bf \em unproblematic} if $H$ has a common vertex  with another forbidden subgraph $H'=K_{t,t}$ of $G$ such that $w(H)\leq w(H')$.

Otherwise, we say that $H$ is {\bf \em problematic}.
\end{definition}


\begin{lemma}\label{lem:cliques-disjoint-problematic}
Any two different problematic subgraphs of $G$ are vertex-disjoint.
\end{lemma}
\dowod
Consider any two different forbidden cliques $H_1$ and $H_2$ with a common vertex. We show that at least one of $H_1, H_2$ is unproblematic.

If both $H_1$ and $H_2$ are \kt{}'s and  $w(H_1)\leq w(H_2)$, then $H_1$ is unproblematic.
Similarly, if both $H_1$ and $H_2$ are \ktt{}'s  and  $w(H_1)\leq w(H_2)$, then $H_1$ is unproblematic.

If $H_1=$\kt{} and $H_2=$ \ktt{}, then by  Point~\ref{itm:cliques-unproblematic-kt-3} of Definition~\ref{def:cliques-unproblematic-kt}, $H_1$ is unproblematic.
\koniec

To find a minimum weight covering co-$t$-matching of $G$, we use the graph $G'$ from Subsection \ref{sec:disjoint-forbidden-cliques}. However, we only replace 
every problematic subgraph by its gadget. In the such constructed $G'$ we compute an \lbmatching{}  of $G'$, which corresponds to a co-$t$-matching $\mc$ of $G$. The co-$t$-matching $\mc$ may not cover some unproblematic forbidden subgraphs. Below we show that it is possible to modify $\mc$ so that it covers every forbidden subgraph.

\begin{lemma}\label{lem:cliques-removing}
Let $\mc$ be any co-$t$-matching of $G$. If $\mc$ does not cover some unproblematic subgraph $H$, then we can find a co-$t$-matching $\nc$ of $G$ such that $w(\nc) \leq w(\mc)$ which covers $H$ and every forbidden subgraph covered by $\mc$.
\end{lemma}
\dowod If $H$ is a \kt{} of $G$ sharing a vertex with some $H'=K_{t,t}$ of $G$, then by Lemma~\ref{lem:cliques-kt-ktt-description}, $t=3$, $H=K_4, \ H'=K_{3,3}$ of $G$ and $H\cap H'=K_{2,2}$.
Let $v_1 \in V(H)\cap V_1(H'), \ v_2 \in V(H)\cap V_2(H'), \ u_1\in V_1(H')\setminus V(H)$ and $u_2\in V_2(H')\setminus V(H)$.  We set $\nc=\mc\setminus\{(v_1,u_2),(v_2,u_1)\}\cup\{(v_1,v_2),(u_1,u_2)\}$.  Notice that the only edge of $H'$ incident to $u_1$ or $u_2$ which can possibly not belong to $\mc$ is $(u_1,u_2)$ since every vertex of $H$ is incident to $t=3$ edges of $H$ which do not belong to $\mc$. Thus, both $u_1$ and $u_2$ are incident to some edge of $\nc$. Moreover, $w(\nc)\leq w(\mc)$ because $w$ is vertex-induced on $H'$.

If $H$ is a \kt{} of $G$ which does not share a vertex with any \ktt{} of $G$, then $H$ shares a vertex with another $H'=K_{t+1}$ of $G$ such that $w(H)\leq w(H')$.
Let $u$ (resp. $u'$) be the only vertex of $H$ (resp. $H'$) which does not belong to $H'$ (resp. $H$). Notice that every edge of $H'$ incident to $u'$ belongs to $\mc$,  so $\deg_{\mc}(u')\geq t$. Since $w(H)\leq w(H')$, the sum of the weights of edges of $H$ incident to $u$ is not greater than the sum of weights of edges of $H'$ incident to $u'$. Therefore, there exists a common vertex of $H$ and $H'$, say $z$, such that $w(u,z)\leq w(u',z)$. We define $\nc$ as $\mc$ with the edge $(u',z)$ replaced by $(u,z)$. As a result, $\nc$ covers $H$. It still covers $H'$ since every vertex of $V(H)\cap V(H')$ different than $z$ is matched to $u'$ in $\nc$. Notice that $\nc$ is still a co-$t$-matching since $\deg_{\nc}(u')=\deg_{\mc}(u')-1 \geq t-1 \geq 1$ since $t \geq 3$.

Let $H$ be a \ktt{} of $G$ sharing a vertex with another $H'=K_{t,t}$ of $G$ such that $w(H)\leq w(H')$. If  $H\cap H'=K_{t-1,t-1}$, we proceed analogously as in the case when $H=\kt{}$ shares a vertex with $H'=\ktt{}$.
Otherwise, $H \cap H'=K_{t,t-1}$ and we proceed analogously as in the case when $H=\kt{}$ shares a vertex with another $H'=\kt{}$. \koniec

\tikzset{
	partite3/.pic=
	{
		\tikzmath{\a=4;\b=25;}
		\foreach \i in {0,...,2}
		{
			\node (v0\i)	at (\i*120-\b:\a)	[bigVertex]{};
			\node (v2\i)	at (\i*120+\b:\a)	[bigVertex]{};
			\node (v1\i)	at ($(v0\i)!0.5!(v2\i)$) [bigVertex]{};
			\foreach \j in {0,...,2}
				\draw[e0] (v\j\i) -- +(\i*120:\a/3);
		}
		\foreach \i in {0,...,2}
		{
			\tikzmath{int \ii; \ii=Mod(\i+1,3);}
			\foreach \j in {0,...,2}
				\foreach \jj in {0,...,2}
					\draw[e0] (v\j\i) -- (v\jj\ii);
		}
	},
	partite3Gadget/.pic=
	{
		\tikzmath{\a=7;\b=35;}
		\foreach \i in {0,...,2}
		{
			\node (v0\i)	at (\i*120-\b:\a)	[bigVertex]{};
			\node (v2\i)	at (\i*120+\b:\a)	[bigVertex]{};
			\node (v1\i)	at ($(v0\i)!0.5!(v2\i)$) [bigVertex]{};
		}
		\node (z) at (0,0) [vertex,label={[font=\Huge]left:$z_H$}]{};
		\foreach \i in {0,...,2}
		{
			\node (u\i) at ($(v1\i)!0.5!(z)$) [vertex]{};
			\draw[e0] (z) -- (u\i);
		}
		\node also [label={[font=\Huge,xshift=10]below left:$u_H^3$}] (u0);
		\node also [label={[font=\Huge,xshift=3]right:$u_H^1$}] (u1);
		\node also [label={[font=\Huge,xshift=3]right:$u_H^2$}] (u2);
		\foreach \i in {0,...,2}
			\foreach \j in {0,...,2}
				\draw[e0] (v\j\i) -- (u\i);
	},
	partite2/.pic=
	{
		\tikzmath{\a=5;\b=20;}
		\foreach \i in {0,...,3}
		{
			\node (v0\i)	at (45+\i*90-\b:\a)	[bigVertex]{};
			\node (v1\i)	at (45+\i*90+\b:\a)	[bigVertex]{};
			\ifnum \i < 2
				\foreach \j in {0,...,1}
					\draw[e0] (v\j\i) -- +(45+\i*90:\a/3);
			\else
				\draw[e0] (v0\i) -- (v1\i);
			\fi
		}
		\node also [label={[font=\Huge,yshift=-0.2]below left:$c_H$}] (v02);
		\foreach \i in {0,...,3}
			\foreach \ii in {0,...,3}
			{
				\ifnum \i < \ii
					\foreach \j in {0,1}
						\foreach \jj in {0,1}
							\draw[e0] (v\j\i) -- (v\jj\ii);
				\fi
			}
	},
	partite2Gadget/.pic=
	{
		\tikzmath{\a=7;\b=20;}
		\foreach \i in {0,...,3}
		{
			\node (v0\i)	at (45+\i*90-\b:\a)	[bigVertex]{};
			\node (v1\i)	at (45+\i*90+\b:\a)	[bigVertex]{};
		}
		\node (z) at (0,0) [vertex,label={[font=\Huge]below right:$z_H$}]{};
		\foreach \i in {0,1}
		{
			\node (u\i)	at (45+\i*90:4)	[vertex]{};
			\draw[e0] (v0\i) -- (u\i) -- (z);
			\draw[e0] (v1\i) -- (u\i);
		}
		\node also [label={[font=\Huge,xshift=10]below left:$u_H^1$}] (u1);
		\node also [label={[font=\Huge,xshift=10]below:$u_H^2$}] (u0);
		\node also [label={[font=\Huge,yshift=-0.2]below:$c_H$}] (v02);
		\node (uc) at ($(z)!0.5!(v02)$) [vertex,label={[font=\Huge]below right:$u_H^c$}]{};
		\draw[e0] (z) to[out=180,in=45] (uc) to[out=180,in=45] (v02);
		\draw[e0] (z) to[out=235,in=0] (uc) to[out=235,in=0] (v02);
	}
}


\section{The bounded \kpq{}-free $t$-matching  problem} \label{sec:partite}
In  this section we consider the weighted bounded \kpq{}-free $t$-matching problem. Recall  that $t=(p-1)q$ for this problem. Let us observe that for $q=1$, $K^p_q=K_{t+1}$ and for $p=2$, $K^p_q=K_{t,t}$. Since the bounded  \kpq{}-free $t$-matching problem for these cases can be solved by a slight modification of the algorithm for the bounded  restricted $t$-matching problem, we assume that $q\geq 2$ and $p\geq 3$. Recall that we assume  that the weight function $w$ is vertex-induced on every subgraph of $G$ isomorphic to \kpq{}.

We denote the empty graph on $n$ vertices by $I_n$. For graphs $G_1=(V_1,E_1)$ and $G_2=(V_2,E_2)$ such that $V_1$ and $V_2$ are disjoint, we say that a graph $G_1\times G_2=(V_1\cup V_2,E_1\cup E_2 \cup (V_1\times V_2))$ is a {\bf \em product} of $G_1$ and $G_2$.

Similarly  as in Section~\ref{sec:cliques}, we build a graph $G'=(V', E')$ together with a weight function \mbox{$w':E'\to\mathbb{R}$}, in which some \kpq{}'s of $G$ are augmented with gadgets.


\subsection{Case $q\geq 3$}
The case of $q\geq 3$ is a little simpler than that of $q=2$. We consider it first.

Let $H$ be any \kpq{} of $G$. We define a {\bf \em gadget for $H$}. We introduce $p$ subdivision vertices $u_H^1$, $u_H^2$, \ldots, $u_H^p$ and a {\bf \em global vertex} $z_H$. For every $i$, we connect $u_H^i$ with every vertex of $V_i(H)$ by half-edges and with $z_H$ (see Fig.~\ref{fig:partiteGadget}).
Let $r_H$ be any potential function of $H$.
A half-edge of $H$ incident to vertex $v_i$ gets weight $r_H(v_i)$. The remaining edges of the gadget get weight $0$.
We set the capacity interval of vertex $z_H$ to $[p-2,p-2]$ and of every subdivision vertex of the gadget for $H$ to $[1,1]$.


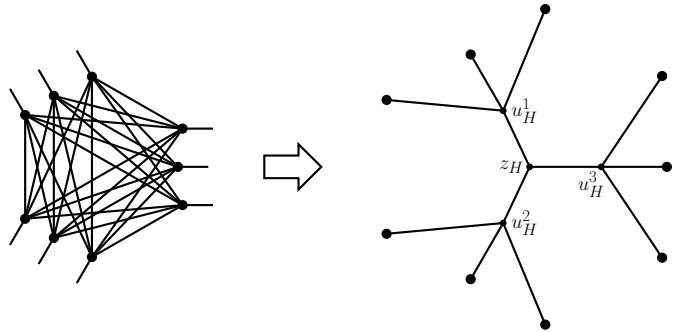
\begin{figure}[htpb]
\centering
\begin{tikzpicture}[scale=0.3,transform shape]
	\pic (arrow) [scale=0.5]{transformsTo};
	\pic (H) [left=7cm of arrow-leftEnd] {partite3};
	\pic (g) [right=9cm of arrow-rightEnd] {partite3Gadget};
\end{tikzpicture}
\caption{A gadget for  \kpq{} for $p=3$ and $q=3$.\label{fig:partiteGadget}}
\end{figure}


\begin{lemma}\label{lem:partite-description}
Let $H_1$ and $H_2$ be any two different \kpq{}'s of $G$ with a common vertex. Then $H_1 \cap H_2$  contains $I_{q-1} \times K^{p-1}_q$.
\end{lemma}
\dowod
Let $V_1$, $V_2$, \ldots, $V_p$ be the color classes of $H_1$ and $v\in V_1$ a  common vertex of $H_1$ and $H_2$.
We first show that at most one vertex of $H_1$ does not belong to $H_2$.
 Since $v$ is adjacent to every vertex of $V_2$, $V_3$, \ldots, $V_p$ via $(p-1)q=t$ edges and  $\deg_G(v)\leq t+1$, there exists at most one other edge incident to $v$. Exactly $t$ of all edges incident to $v$ belong to $H_2$, therefore all vertices of $V_2$, $V_3$, \ldots, $V_p$ except at most one belong to $H_2$ as well. Suppose that if there exists a vertex  of $H_1$ which does not belong to $H_2$, then it belongs to $V_p$. Let $u$ be any vertex of $V_2$.  Since $p\geq 3$, $u \notin V_p$ and hence $u \in H_2$. By  repeating the same argument as for $v$, we get that all vertices of $V_1$, $V_3$, $V_4$, \ldots, $V_p$ except at most one belong to $H_2$, which means that all vertices of $V_1$ belong to $H_2$ and hence at most one vertex of $H_1$ does not belong to $H_2$.

\begin{claim}\label{claim:partite-classes}
If all vertices of one color class $V_i$ belong to $H_2$, then they belong to the same color class of $H_2$.
\end{claim}
\dowod
Suppose to the contrary  that two vertices $v_1,v_2\in V_i$ belong to different color classes of $H_2$. This means that $H_2$ contains an edge $(v_1,v_2)$.  Let $v_3$ be any other vertex of $V_i$ (it exists since $q\geq 3$). It cannot belong to the same color class of $H_2$ as both $v_1$ and $v_2$.  Thus it is adjacent to at least one of $v_1, v_2$. Suppose that $v_3$ is a neighbour of $v_1$.  This means that $v_1$ is incident to $t+2$ edges of $G$: $t$ edges of $H_1$, $(v_1,v_2)$ and $(v_1,v_3)$ - a  contradiction. 
\koniec
By the above claim  $H_1$ and $H_2$ have different vertex sets because otherwise they  would be equal.  We have already assumed that if $H_1$ contains a vertex not contained in $H_2$, then it belongs to $V_p$. Thus by the above claim, we conclude that $V_1, V_2, \ldots, V_{p-1}$ are color classes of $H_2$  and $q-1$ vertices of $V_p$ belong to $H_2$, which  
proves  that  $H_1 \cap H_2 = I_{q-1} \times K^{p-1}_q$.
\koniec

\begin{definition}\label{def:partite-problematic}
Let $H$ be any \kpq{} of $G$. We say that $H$ is {\bf \em unproblematic} if there exists another $H'=K^p_q$ of $G$ that shares a vertex with $H$ and such that  $w(H)\leq w(H')$.
\end{definition}


\begin{lemma}\label{lem:partite-disjoint-problematic}
Any two different problematic \kpq{}'s of $G$ are vertex-disjoint.
\end{lemma}
\dowod
Consider any two different \kpq{}'s $H_1$ and $H_2$ of $G$ with a common vertex. By Lemma~\ref{lem:partite-description}, $H_1 \cap H_2$  contains $I_{q-1} \times K^{p-1}_q$. Assume that $w(H_1)\leq w(H_2)$. Then $H_1$ fulfills Definition~\ref{def:cliques-unproblematic-kt}.
\koniec


\begin{theorem}\label{thm:partite}
For any \kpq{}-covering co-$t$-matching $\mc$ of $G$ there exists an \lbmatching{} $M'$ of $G'$ such that $w'(M') = w(\mc)$. \\
For any \lbmatching{} $M'$ of $G'$ there exists a \kpq{}-covering co-$t$-matching~$\mc$ of $G$ such that $w(\mc) \leq w'(M')$.
\end{theorem}

\dowod
We prove the first sentence of Theorem~\ref{thm:partite} in the beginning. 
Let $H$ be any \kpq{} of $G$. Since $\mc$ covers $H$, some edge of $H$ belongs to $\mc$. If more than one edge of $H$ belongs to $\mc$, we choose one of them. Suppose that we chose $(a,b)\in\mc$. Assume that $a\in V_1(H)$ and $b\in V_2(H)$. Then we add edges $(a,u_H^1)$ and $(b,u_H^2)$ to $M'$. Additionally, we add $p-2$ edges, $(z_H,u_H^i)$ for every $i>2$, to $M'$. We add every other edge of $H$ to $M'$.

We prove the second sentence of Theorem~\ref{thm:partite} now. 
Let $H$ be any \kpq{} of $G$. Notice that exactly two subdivision vertices of the gadget for $H$ are matched to vertices of $G$ in $M'$, because $z_H$ is matched to exactly $p-2$ other subdivision vertices. Assume that $u_H^1$ and $u_H^2$ are not matched to $z_H$ in $M'$. Therefore, they are matched to some vertices $a$ and $b$, respectively, of $G$. We add an edge $(a,b)$ and every other edge of $H$ which belongs to $M'$ to $\mc$. Hence, $\mc$ covers every problematic \kpq{} of $G$.

However, $\mc$ may not cover some unproblematic \kpq{}'s of $G$. Let $H$ by any such \kpq{}. From Definition~\ref{def:partite-problematic}, there exists a $H'=K^p_q$ of $G$ such that $H\cap H'=I_{q-1}\times K^{p-1}_q$ and $w(H)\leq w(H')$. Let $u$ (resp. $u'$) be only vertex of $H$ (resp. $H'$) which does not belong to $H'$ (resp. $H$). Rest of the proof is identical to covering unproblematic \kt{} which shares some vertex with another \kt{} of $G$ in the proof of Lemma~\ref{lem:cliques-removing}.
\koniec




\subsection{Case $q=2$}
Next  we consider the case of $q=2$. Notice that $t=2p-2$ in this case. It turns out that  two different \kpt{}'s of $G$ can have the same set of vertices.


\begin{lemma}\label{lem:partite-2-description}
Let $H_1$ and $H_2$ be  two  \kpt{}'s of $G$ with a common vertex. Then either:
\begin{enumerate}
\item\label{itm:partite-2-description-1} $H_1\cap H_2=I_1\times K^{p-1}_2$, or
\item\label{itm:partite-2-description-2}  $H_1$ and $H_2$ have the same set of vertices and at least two color classes of $H_1$ induce an edge in $G$.
\end{enumerate}
\end{lemma}
\dowod
Observe that the first paragraph of the proof of Lemma~\ref{lem:partite-description} holds  in this case because we did not  assume that $q\geq 3$ there. Therefore, $H_1 \cap H_2$   contains $I_{q-1} \times K^{p-1}_q$. If $H_1 \cap H_2 = I_{q-1} \times K^{p-1}_q$, we are done.

If  $H_1$ and $H_2$ have the same set of vertices, we need to show that for some two different color classes of $H_1$, its vertices are connected via an edge of $G$. Since $H_1$ and $H_2$ are different, there exists a color class of $H_1$ consisting of vertices $v_1$ and $v_2$ which belong to different color classes of $H_2$. Since an edge $(v_1,v_2)$ belongs to $H_2$,  the color class of $H_1$ consisting of $v_1$ and $v_2$ is one of the color classes we are looking for. Let $u$ be the other vertex which belongs to the same color class of $H_2$ as $v_1$.  The color class of $H_1$  to which $u$ belongs is the second class from the claim.
\koniec

\begin{definition}\label{def:partite-dense}
Let $A$ be any subset of $2p$ vertices of $G$. We say that $G[A]$ is {\bf \em dense} if it contains at least two different forbidden \kpt{}'s.
The set $C_{G[A]}=\{ a \in A: N(a) \subset A\ \wedge |N(A)|=t+1\}$ is called the {\bf \em core} of $G[A]$.
\end{definition}


\begin{lemma}\label{lem:partite-core}
The size of the core of any dense subgraph of $G$ is even and at least  four.
\end{lemma}
\begin{proof}
Consider two different \kpt{}'s $H_1$ and $H_2$ with the same vertex set $A$. Notice that the core of $A$ is equal to the sum of color classes of $H_1$ whose induce an edge in $G$, so its size is even. Moreover, its size is at least four, because by Lemma~\ref{lem:partite-2-description}, there are at least two such color classes.
\end{proof}


To guarantee that every \kpt{} contained in a dense subgraph of $G$ is covered by the computed co-$t$-matching, we introduce a separate gadget for a dense subgraph. These gadgets are used alongside  gadgets for problematic \kpt{}'s. (We generalize Definition~\ref{def:partite-problematic} and the gadget for \kpq{} to the case $q=2$ in a straightforward manner.) In the following lemma we prove that any dense subgraph of $G$ is vertex-disjoint both with any other dense subgraph and with any  other problematic \kpt{}'s of $G$. 
\begin{lemma}\label{lem:dense-disjoint}
Let $H$ be any dense graph of $G$. Then any \kpt{} of $G$ which shares a vertex with $H$ has  the same set of vertices as $H$.
\end{lemma}
\dowod
Let $H'$ be any \kpt{} of $G$ which has a common vertex with $H$. Suppose that the set of vertices of $H'$ is different than $V(H)$. Since there exists some \kpt{} of $H$,  we can use Lemma~\ref{lem:partite-2-description} and conclude that all but one vertex of $H'$ belongs to $H$. Let $v\in  V(H') \setminus V(H)$. Notice that vertex $v$ is adjacent to at least $2p-2$ vertices of $H$. By Lemma~\ref{lem:partite-core}, the core of $H$ consists of at least four vertices, which means that at least two of these vertices are adjacent to $v$ - a contradiction with the definition of the core.
\koniec

The following observation describes all \kpt{}'s contained in a dense subgraph of $G$.

\begin{observation}\label{obs:all-kpt-in-dense-subgraph}
Consider any dense subgraph $H$ of $G$ with the core $C_H$. Let $H'$ be any \kpt{} of $H$. Then the color classes of any \kpt{} of $H$ consist of the color classes of $H'$ which are not contained in $C_H$ and some partition of $C_H$ into two element subsets. Moreover, any such partition of $C_H$ forms the color classes of exactly one \kpt{} of $H$. 
\end{observation}


\begin{lemma}\label{lem:partite-dense-potential}
The  weight function $w$ is vertex-induced on any dense subgraph $H$ of $G$. Moreover, the potential function of any \kpt{} of $H$ is also a potential function $r_H$ of $H$.
\end{lemma}
\dowod
We show that the potential functions of any two different \kpt{}'s of $H$ are the same. It ends proof since, by Observation~\ref{obs:all-kpt-in-dense-subgraph}, every edge of $H$ belongs to some \kpt{} of $H$.

Let $H_1$ and $H_2$ be any two different \kpt{}'s of $G$. Let $r_1$ and $r_2$ be any potential funcions of $H_1$ and $H_2$, respectively. Consider any vertex $v$ of $H$. Assume that $v\in V(H_1)$. We show that there exists a triangle $(v,v_1,v_2)$ which belongs to both $H_1$ and $H_2$. In the light of Observation~\ref{obs:triangle-potentials} it implies that $r_1(v)=r_2(v)$.

Set $v_1$ to any vertex of $V_2(H_1)$ which belongs to other color class of $H_2$ than $v$. If there exists some vertex of $V_3(H_1)$ (this color class exists since $p\geq 3$) which belongs to the same color class of $H_2$ as neither $v$ nor $v_1$, we choose it as $v_2$.

Otherwise, both vertices of $V_3(H_1)$ belong to the same color class of $H_2$ as $v$ and $v_1$, respectively. In this case we choose $v_2$ in such a way that it belongs to the same color class of $H_2$ as $v_1$ and we change $v_1$ to another vertex of $V_2(H_1)$. Since $q=2$, $v_1$ has to be in different color class of $H_2$ than $v$ and $v_2$.
\koniec


To construct a gadget for a dense subgraph of $G$ we make use of the following lemma. In our applications, $M$ will typically be a set of edges of a co-$t$-matching of $G$ which belong to $H$.

\begin{lemma}\label{lem:partite2-edges}
Let $H$ be any dense subgraph of $G$ with the core $C_H$ of size $2k$. Assume that every vertex of $C_H$ is incident to at least one edge of a subset $M \subseteq E(H)$. Then the following statements are equivalent.
\begin{enumerate}
\item $M$ covers every \kpt{} of $H$.
\item $M$ is not a perfect matching of $G[C_H]$.
\item some vertex of $C_H$ is incident to at least two edges of $M$, or some edge of $M$ is incident to some vertex of $V(H)\setminus C_H$.
\end{enumerate}
\end{lemma}
\dowod
We prove the three following implications.

\emph{(1) $\Rightarrow$ (2).}
We prove it by contraposition. Assume that $M$ is a perfect matching of $G[C_H]$. Then $M$ forms a partition of some $H'=\kpt{}$ of $H$ together with color classes of any \kpt{} of $H$ whose are not contained in $C_H$. Notice that $H'$ is not covered by $M$.

\emph{(2) $\Rightarrow$ (3).}
We prove it by contraposition. Assume that every vertex of $C_H$ is incident to at most one edge of $M$ and every edge of $M$ has both endpoints in $C_H$. From our assumption, every vertex of $C_H$ is incident to exactly one edge of $M$. Therefore, $M$ is a perfect matching of $G[C_H]$.

\emph{(3) $\Rightarrow$ (1).}
Let $H'$ be any \kpt{} of $H$. We will show that some edge of $M$ belongs to $E(H')$. We consider two cases. If some vertex $v$ of $C_H$ is incident to two different edges of $M$, say $(v,u_1)$ and $(v,u_2)$, then at least one of these edges belongs to $E(H')$ since $u_1$ or $u_2$ belongs to a different color class of $H'$ than $v$. If some edge $e$ of $M$ is incident to some vertex of $V(H)\setminus C_H$, then $e$ belongs to $E(H')$ since, by Observation~\ref{obs:all-kpt-in-dense-subgraph}, it belongs to every \kpt{} of $H$. 
\end{proof}

Let $H$ be any dense subgraph of $G$ with the core~$C_H$ of size $2k$. Let $V_1, V_2,\ldots, V_{p-k}$ be the color classes of any \kpt{} of $H$ whose are not contained in $C_H$. 
We define $G_H$ -- {\bf \em a gadget for $H$}. By Lemma~\ref{lem:partite-dense-potential}, $w$ is vertex-induced on $H$. Let $r_H$ be the potential function of $H$. We define $c_H$ -- {\bf \em a center of $H$} -- as a vertex of $C_H$ of the minimum potential. If there many such vertices, we choose any of them. If the potential of $c_H$ is negative, then the gadget for $H$ is empty. We show later in the proof of Theorem~\ref{thm:partite-2-matching-to-opt} that every \kpt{} of such dense subgraph $H$ is covered by the constructed co-$t$-matching $\mc$.

We assume now that the potential of $c_H$ is nonnegative. In such a case, we introduce $p-k+1$ subdivision vertices $u_H^c,u_H^1,u_H^2,\ldots,u_H^{p-k}$ and a global vertex $z_H$. For every $i$, we connect $u_H^i$ with every vertex of $V_i$ by half-edges and with $z_H$. We connect $u_H^c$ with $c_H$ by two half-edges and with $z_H$ by two edges (see Fig.~\ref{fig:partite2Gadget}). A half-edge of $H$ incident to vertex $v$ gets weight $r_H(v)$. Notice that both half-edges incident to $c_H$ get weight $r_H(c_H) \geq 0$. The remaining edges of the gadget get weight $0$. We set the capacity interval of vertex $z_H$ to $[p-k,p-k]$, of vertex $u_H^c$ to $[2,2]$ and of every other subdivision vertex of the gadget for $H$ to $[1,1]$.

The main idea behind this gadget is the following. In the vein of Lemma~\ref{lem:partite2-edges}, we want to ensure that some vertex of $C_H$ is incident to some two edges of $\mc$ or there exists an edge of $\mc \cap E(H)$ incident to some vertex of $V(H)\setminus C_H$. If some vertex of $C_H$ is incident to at least two edges of $\mc$, then we can assume that it is a center of $H$ since $\mc$ is not minimum otherwise. If some edge of $\mc$ has exactly one endpoint in $C_H$, then we can assume that it is $c_H$ by replacing some edges of $\mc$ if necessary. Therefore, it is sufficient to guarantee that exactly two half-edges of $H$ are present in $M'$. If both of these half-edges are adjacent to $c_H$, it corresponds to including some two edges incident to $c_H$ in $\mc$ by "splitting" some edge of $M'$ whose both endpoints belong to $C_H$. Otherwise, these half-edges correspond to some edge of $H$ with some endpoint in $V(H)\setminus C_H$.

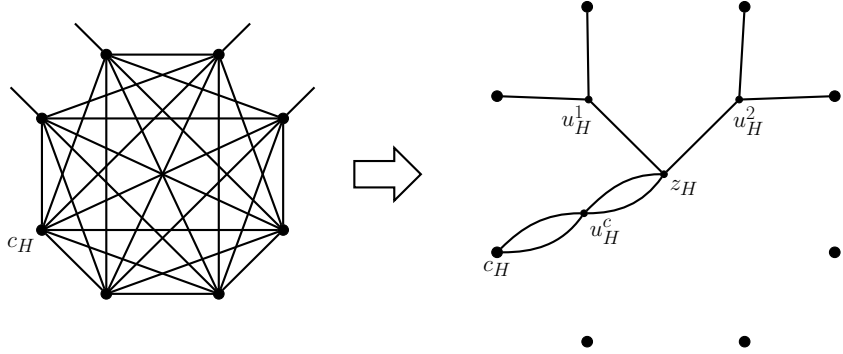
\begin{figure}[htpb]
\centering
\begin{tikzpicture}[scale=0.35,transform shape]
	\pic (arrow) [scale=0.5]{transformsTo};
	\pic (H) [left=7cm of arrow-leftEnd] {partite2};
	\pic (g) [right=9cm of arrow-rightEnd] {partite2Gadget};
\end{tikzpicture}
\caption{A gadget for a dense subgraph consisting of eight vertices with the core of size four whose center $c_H$ has nonnegative potential.\label{fig:partite2Gadget}}
\end{figure}


\begin{theorem}\label{thm:partite-2-opt-to-matching}
Let $\mc$ be a minimum weight \kpq{}-covering co-$t$-matching of $G$. Then we can find an \lbmatching{} $M'$ of $G'$ such that $w'(M') = w(\mc)$.
\end{theorem}
\dowod Let $H$ be any dense subgraph of $G$ with the core $C_H$ of size $2k$. Define $\mc_H=\mc\cap E(H)$. We add every edge of $\mc_H$ to $M'$. If the gadget for $H$ is empty, we are done. Therefore, we can assume that the potential of the center of $H$ is nonnegative. The main idea is to find an edge $e=(v_1,v_2)$ of $\mc_H$ such that $v_1,v_2\in V(H)\setminus C_H\cup\{c_H\}$ and to replace it in $M'$ by two half-edges -- one incident to $v_1$ and one incident to $v_2$. Then we add $p-k$ edges incident to $z_H$ to $M'$ in order to satisfy the degree contraints of the subdivision vertices of the gadget for $H$. If such edge $e$ exists, we are done. Therefore, we assume that such edge does not exist.

If some vertex $v\in C_H\setminus\{c_H\}$ is matched to some vertex $u\in V(H)\setminus C_H$, we consider two cases. If $v$ is incident to another edge of $\mc_H$ besides $(v,u)$, we change $\mc$ by replacing $(v,u)$ by $(c_H,u)$. Notice that this replacement does not increase the weight of $\mc$ since $c_H$ has the minimum potential among the vertices of $C_H$. Moreover, $\mc$ is still a co-$t$-matching of $G$ after this operation because $v$ is still incident to some edge of $\mc$. If $(v,u)$ is the only edge of $\mc_H$ incident to $v$, let $v'$ be any vertex of $C_H$ matched to $c_H$ in $\mc_H$. Notice that $v'$ exists since $\mc$ is a co-$t$-matching of $G$. From our assumption, $v\neq v'$ and $(v,v')\notin\mc$. We replace $(v,u)$ and $(c_H,v')$ in $\mc$ by $(c_H,u)$ and $(v,v')$. Notice that this replacement does not affect the weight of $\mc$ because, by Lemma~\ref{lem:partite-dense-potential}, $w$ is vertex-induced on $H$. In both cases, we choose $(c_H,u)$ as the edge $e$.

If no vertex of $C_H\setminus\{c_H\}$ is matched to a vertex of $V(H)\setminus C_H$, then every edge of $\mc_H$ has both endpoints in $C_H$. Recall that $\mc$ is \kpq{}-covering, so by Lemma~\ref{lem:partite2-edges}, some vertex of $C_H$ is incident to two edges of $\mc_H$. By Lemma~\ref{lem:partite-dense-potential}, $w$ is vertex-induced on $H$, so $w(\mc_H)=\sum_{v\in C_H}r_H(v)\deg_{\mc_H}(v)$. The main idea is to replace $\mc_H$ in $\mc$ by a subset $N$ of $k+1$ edges of $G[C_H]$ in which $c_H$ has degree at least two in such a way that only the vertices of smallest potential in $C_H$ have degree higher than one in $N$. It guarantees that $w(N)\leq w(\mc_H)$ since, from our assumption, the potentials of the vertices of $C_H$ are nonnegative. We consider two cases. If some vertex of $C_H$ is incident to at least three edges of $\mc_H$, then $N$ consists of three different edges of $G[C_H]$ incident to $c_H$ and a perfect matching of a subgraph induced by the remaining $k-4$ vertices of $C_H$. Since the degree of $c_H$ in $N$ is three and the degree of every other vertex of $C_H$ in $N$ is one, we can conclude that $w(N)\leq w(\mc_H)$. Otherwise, some two vertices of $C_H$ have degree two in $\mc_H$. Let $c_H'$ be the vertex of $C_H\setminus\{c_H\}$ of the minimum potential. Then $N$ consists of an edge $(c_H,c_H')$ and any perfect matching of $G[C_H]$ which does not contain edge $(c_H,c_H')$. Since the degrees of both $c_H$ and $c_H'$ in $N$ are two and the degree of every other vertex of $C_H$ in $N$ is one, we can conclude that $w(N)\leq w(\mc_H)$. In both cases, we remove $\mc_H$ from $M'$ and we choose any two different edges $(c_H,v_1)$ and $(c_H,v_2)$ of $N$. We add every other edge of $N$ to $M'$. Moreover, we add an edge $(v_1,v_2)$ and both half-edges incident to $c_H$ to $M'$. At the end, we add every edge of the gadget for $G_H$ which connects $z_H$ and any subdivision vertex different than $u_H^c$ to $M'$. Notice that there are exactly $p-k$ such edges.
\koniec


\begin{theorem}\label{thm:partite-2-matching-to-opt}
Let $M'$ be a minimum weight \lbmatching{} of $G'$. Then we can find a \kpq{}-covering co-$t$-matching~$\mc$ of $G$ such that $w(\mc) \leq w'(M')$.
\end{theorem}
\dowod
We assume that $M'$ has a minimum number of edges among all minimum weight \lbmatching{}s of $G'$.
We initialize $\mc$ as the empty set. We add every edge of $M'$ which belongs to $G$ to $\mc$. We consider problematic \kpt{}'s of $G$ analogously to \kpq{}'s in the proof of Theorem~\ref{thm:partite}.

At the beginning we prove that $\mc$ covers every \kpt{} contained in any dense subgraph of $G$ whose center has negative potential. Suppose that $\mc$ does not cover some \kpt{} of a dense subgraph $H$ with the core $C_H$ such that $r_H(c_H)<0$. Notice that every vertex of $C_H$ is incident to some edge of $\mc_H=\mc\cap E(H)$ because it has degree $t+1$ in $G$ and in $G'$. Therefore, we can use Lemma~\ref{lem:partite2-edges} and conclude that $\mc_H$ is a perfect matching of $G[C_H]$. From the construction of $\mc$, $M_H'=M'\cap E(H)$ is a perfect matching of $G[C_H]$ as well. Define $M''=M'\setminus M'_H\cup\{(c_H,v):v\in C_H\setminus\{c_H\}\}$. Notice that $M''$ is an \lbmatching{} of $G'$ such that $w'(M'') < w'(M')$ since $r_H(c_H)<0$. A contradiction with the minimality of $M'$.

Consider any dense subgraph $H$ of $G$ with the core $C_H$ whose center has nonnegative potential. From the construction of the gadget for $H$, exactly two its half-edges belong to $M'$. If these half-edges are incident to two different vertices $u_1$ and $u_2$ of $G$, then we add an edge $(u_1,u_2)$ to $\mc$ unless it belongs already. Notice that $u_1$ or $u_2$ belongs to $V(H)\setminus C_H$, so by Lemma~\ref{lem:partite2-edges}, $\mc$ covers all \kpt{}'s of $H$.
Otherwise, both half-edges of the gadget for $H$ which belong to $M'$ are incident to $c_H$.

If there exists an edge $(v_1,v_2)$ of $M'$ such that both $v_1$ and $v_2$ belong to $C_H\setminus\{c_H\}$, we remove $(v_1,v_2)$ from $\mc$. Then we add edges $(c_H,v_1)$ and $(c_H,v_2)$ to $\mc$ unless they belong to $\mc$ already. From now we assume that there is no such edge.

We claim that at most one vertex of $C_H\setminus\{c_H\}$ is matched to $c_H$ in $M'$. Suppose that some two different edges $(c_H,v_1)$ and $(c_H,v_2)$ belong to $M'$ where $v_1$ and $v_2$ belong to $C_H$. From our assumption, $(v_1,v_2)$ does not belong to $M'$, so we can construct an \lbmatching{} of $G'$ of weight smaller than $w'(M')$ by replacing $(c_H,v_1)$ and $(c_H,v_2)$ in $M'$ by $(v_1,v_2)$. A contradiction with the minimality of $M'$. Notice that we use the assumption that $r_H(c_H) \geq 0$ here.

Since every vertex of $C_H\setminus\{c_H\}$ is incident to at least one edge of $M'$ and the core of $H$ consists of at least four vertices, there exist two edges $(v_1,u_1)$ and $(v_2,u_2)$ which belong to $M'$ such that $v_1$ and $v_2$ belong to $C_H$ whereas $u_1$ and $u_2$ belong to $V(H)\setminus C_H$. We replace edges $(v_1,u_1)$ and $(v_2,u_2)$ in $\mc$ by edges $(v_1,v_2)$, $(c_H,u_1)$ and $(c_H,u_2)$ unless they belong to $\mc$ already.

\koniec

\section{Running time analysis}\label{sec:running}

In this section we consider the running time of Algorithm~\ref{alg:main} for both the bounded restricted $t$-matching problem and the bounded \kpq{}-free $t$-matching problem.
Since all problematic and dense subgraphs of $G$ are pairwise vertex-disjoint, graph $G'$ constructed in Step~\ref{itm:step1} of Algorithm~\ref{alg:main} has $O(n)$ vertices and $O(m)$ edges indeed. We prove that Step~\ref{itm:step3} of Algorithm~\ref{alg:main} can be implemented to run in time $\mathcal{O}(\min\{nm\log{n},n^3\})$ in the weighted variant (see Subsection~\ref{subsec:running-weighted}) and $\mathcal{O}(\sqrt{n}m)$ in the unweighted (see Subsection~\ref{subsec:running-unweighted}). 

Observe that we have to find all forbidden subgraphs of $G$ in order to build $G'$ and remove unproblematic subgraphs from the acquired co-$t$-matching $\mc$ of $G$ in Step~\ref{itm:step5} of Algorithm~\ref{alg:main}. Moreover, we need to know which forbidden subgraphs of $G$ have common vertices in order to determine which of them are problematic. We prove that this can be done in $\mathcal{O}(m)$ time in Subsection~\ref{sec:running-kpq}.

In the weighted setting, we also have to calculate and compare the weights of forbidden subgraphs of $G$ which are not part of any dense subgraph in order to determine which of them are problematic. Notice that, for any such two different forbidden subgraphs $H$ and $H'$ of $G$ which have a common vertex, we can obtain the weight of $H'$ from the weight of $H$ in $\mathcal{O}(t)$ time because the vertex sets of $H$ and $H'$ differ by at most two vertices. We initially calculate the weights of some vertex-disjoint forbidden subgraphs of $G$ naively and we use these results to obtain the weights of the remaining forbidden subgraphs.

On the other hand, we have to know the potential function of every problematic and dense subgraph of $G$ in order to define the weights of half-edges in the gadgets in Step~\ref{itm:step2} of Algorithm~\ref{alg:main}. Since we do not assume that the forbidden subgraphs are given as a part of the input, we cannot assume that the potential functions of forbidden subgraphs are given. However, we show that the potential function of any forbidden subgraph of $G$ can be extracted from the weight function $w$ in $\mathcal{O}(t^2)$ time in Subsection~\ref{sec:running-potential}. By Lemma~\ref{lem:partite-dense-potential}, the potential function of any dense subgraph $H$ of $G$ can be found by extracting the potential function of any forbidden subgraph contained in $H$. Therefore, the potential functions of all problematic and dense subgraphs of $G$ can be found in $\mathcal{O}(m)$ time because they are vertex-disjoint.

We claim that Step~\ref{itm:step5} of Algorithm~\ref{alg:main} can be implemented to run in $\mathcal{O}(m)$ time since all forbidden subgraphs of $G$ which are not covered by co-$t$-matching $\mc$ are vertex-disjoint. It follows from the simple observation that any two different forbidden subgraph $H_1$ and $H_2$ with a common vertex have some common vertex $v$ of degree exactly $t+1$ whose every incident edge belongs to $H_1$ or $H_2$. Since $\mc$ is a co-$t$-matching of $G$, such vertex $v$ has to be incident to some edge of $\mc$, so $\mc$ covers $H_1$ or $H_2$.

We discuss briefly how $G$ is stored in the memory by Algorithm~\ref{alg:main}. For every vertex of $G$, it stores a list of its neighbours, in no particular order. In the weighted variant, every entry in these lists also contains the weight of the corresponding edge.

\subsection{Computing a minimum weight \lbmatching{} of $G'$ in the weighted case}\label{subsec:running-weighted}
In this subsection, we present how to find a minimum weight \lbmatching{} of $G'$ in $\mathcal{O}(\min\{nm\log{n},n^3\})$ time. We use the algorithm from Theorem~\ref{thm:gabow} for $G'$, $l$ and $b$ with the weight function $-w$, i.e. we set the weight of every edge $e$ of $G'$ to $-w'(e)$. However, this approach results in $\mathcal{O}(\min\{m^2\log{n},n^2m\})$ running time since sum of the upper bounds of the capacity intervals of the vertices of $G'$ is $\mathcal{O}(m)$. We show how to modify the algorithm from Theorem~\ref{thm:gabow} to work in $\mathcal{O}(\min\{nm\log{n},n^3\})$ time in our case. To explain this, we have to briefly present how this algorithm works.

At the beginning, the algorithm constructs an auxiliary graph $G^*$ by creating two copies of graph $G'$ and connecting the corresponding copies $v_1$ and $v_2$ of every vertex $v$ of $G'$ by $b(v)-l(v)$ paths of length three. We call these paths {\bf \em added paths} (of vertex $v$). Both edges incident to endpoints of an added path $P$ are said to be {\bf \em outermost}. The remaining edge of $P$ is said to be {\bf \em middle}. Both endpoints of the middle edge of $P$ are said to be {\bf \em central}. The algorithm calculates a maximum weight $(b^*,b^*)$-matching of $G^*$ where $b^*(v_i)=b(v)$ for every copy $v_i$ of vertex $v$ of $G'$ and $b^*(u)=1$ for every other vertex $u$ of $G^*$. The weight function $w^*$ of $G^*$ is defined as follows. The weight of any copy of an edge $e$ of $G'$ is equal to $-w'(e)$. Let $W=\max_{e\in E(G')}\{|w'(e)|\}$. For every added path, its middle edge gets weight $2W$ and both outermost edges get weight $W$. It is easy to check that a maximum weight \lbmatching{} of $G'$ with the weight function $-w'$ corresponds to a maximum weight $(b^*,b^*)$-matching of $G^*$ with the weight function $w^*$.

The algorithm finds a maximum weight $(b^*,b^*)$-matching of $G^*$ as follows. It simulates the classical algorithm for finding a maximum weight matching, given by Edmonds~\cite{Edmonds1965} and improved by Gabow~\cite{Gabow1975} and Galil, Micali and Gabow~\cite{GalilEtAl1986}, in an auxiliary graph $\gm$. Graph~$\gm$ is constructed from $G^*$ by splitting every its vertex $v$ into $\deg_{G^*}(v)$ vertices called {\bf \em external vertices} (of $v$). The external vertices of $v$ are incident to $\deg_{G^*}(v)$ edges -- each vertex to one edge -- corresponding to all the edges incident to $v$ in $G^*$. With a slight abuse of notation, we identify every edge of $G^*$ with its corresponding edge in $\gm$. Every external vertex of $v$ which is incident to some outermost edge of some added path is said to be {\bf \em improper}. All other external vertices of $v$ are said to be {\bf \em proper}. Futhermore, $\deg_{G^*}(v)-b^*(v)$ additional vertices, called {\bf \em internal vertices} (of $v$) are added to $\gm$. Every internal vertex of $v$ is connected to every external vertex of $v$ by an edge. We refer to a graph induced by all external and internal vertices of $v$ as {\bf \em a substitute} of $v$. Observe that a substitute of $v$ is isomorphic to $K_{d(v),d(v)-b^*(v)}$ where $d(v)=\deg_{G^*}(v)$. The weight function $\wm$ of $\gm$ is defined as follows. Every edge of $\gm$ which corresponds to some edge $e$ of $G^*$ gets weight $w^*(e)$. Every edge incident to some internal vertex gets weight $2W$.

It is well known (see~\cite{Berge1973} for example) that every $b^*$-matching of $G^*$ corresponds to some matching of $\gm$ in which every internal vertex is matched. Moreover, every $(b^*,b^*)$-matching of $G^*$ corresponds to some perfect matching of $\gm$. Hence, the algorithm computes a maximum weight perfect matching of $\gm$. A matching $\mm$ of $\gm$ is said to be {\bf \em extreme} if it has the maximum weight among all matchings of $\gm$ of size $|\mm|$. The algorithm computes extreme matchings of $\gm$ consisting of exactly $k$ edges, for $k=k_0$, $k_0+1$, \ldots, $S$, for some natural number $k_0$, where $S$ is the number of edges of any perfect matching of $\gm$. Observe that an extreme matching of size $S$ is a sought-after perfect matching of $\gm$. Notice that the algorithm does not store $\gm$ and its extreme matchings explicitly since they have too many edges. The algorithm starts with extreme matching $M_0$ which consists of one incident edge per every internal vertex of $\gm$. In every phase, the algorithm increases the size of the current extreme matching by one. Hence, it works in $S-|M_0|$ phases. Every phase can be implemented to run in $\mathcal{O}(\min\{m\log{n},n^2\})$ time. Notice that the number of phases is equal to half of the number of vertices of $\gm$ unmatched in $M_0$. Since $b^*(v)$ external vertices of every vertex $v$ of $G^*$ are unmatched in $M_0$, the number of phases is equal to
$$ \frac{1}{2}\sum_{v\in V(G^*)}b^*(v)= \sum_{v\in V(G')}b(v)+\frac{1}{2}\sum_{v\in V(G')}(b(v)-l(v)) = O(m). $$

The main idea behind the speed-up of the presented algorithm is to initialize it with extreme matching $M_1$ of $\gm$ which has only $\mathcal{O}(n)$ unmatched vertices instead of $M_0$. Such an improvement results in $\mathcal{O}(n)$ phases instead of $\mathcal{O}(m)$, which gives the desired running time. In fact, we choose $M_1$ which is a maximum weight (not necessarily perfect) matching of $\gm$. It guarantees that $M_1$ is an extreme matching. $M_1$ roughly correponds to a matching $N$ of $\gm$ which contains both outermost edges of every added path. Of course, for every internal vertex $v$ of $\gm$, $N$ contains exactly one edge incident to $v$ which is not incident to any other edge of $N$. Notice that the number of vertices of $\gm$ unmatched in $N$ is equal to $2\sum_{v\in V(G')} l(v) = \mathcal{O}(n)$. However, $N$ does not have to be an extreme matching of $\gm$. Observe that the weight of the middle edge of an added path is equal to the total weight of its outermost edges. Therefore, it may happen that replacing outermost edges of $N$ of some added path $P$ by the middle edge of $P$ and some copy of some edge of $G'$ which has the positive weight results in a matching of $\gm$ of size $|N|$ and weight greater than $\wm(N)$. However, recall that we run the algorithm on $G'$ with the weight function $-w'$. Hence, any edge $e$ of $G'$ has positive weight with regard to $-w'$ if and only if $w'(e)$ is negative. Notice that only half-edges of the gadgets in $G'$ may have negative weights since we assume that the weight $w(e)$ of any edge $e$ of $G$ is nonnegative. A half-edge $e$ has negative weight $w'(e)$ if its endpoint which belongs to $G$ has negative potential. Therefore, we include some half-edges in $M_1$ in such a way that the number of vertices of $\gm$ unmatched in $M_1$ is still $\mathcal{O}(n)$.

Now we give the exact construction of $M_1$. We initialize $M_1$ as the empty set. We partition the edge set of $\gm$ into disjoint subsets, called {\bf \em parts} of $\gm$. For every part $E_0$ of $\gm$, we choose a subset $M\subseteq E_0$, called a {\bf \em representation} of $E_0$, such that, for every matching $N$ of $\gm$, $\wm(M)\geq\wm(N\cap E_0)$. We add the representation of every part of $\gm$ to $M_1$. We choose the representations of parts of $\gm$ in such a way that $M_1$ is a matching of $\gm$. It is easy to check that  $M_1$ is indeed a maximum weight matching of $\gm$.

Now we present a partition of the edge set of $\gm$ into parts of $\gm$ and its representations.
\begin{enumerate}
\item
Consider any copy $u_i$ in $G^*$ of a subdivision vertex $u$ of $G'$. Observe that no added path of $u$ was added to $G^*$, so all external vertices of $u_i$ are proper. Notice that the external vertices of $u_i$ are incident to some half-edges and possibly to some edges of weight zero. Furthermore, the substitute of $u_i$ consists of $\deg_{G'}(u)$ external and $\deg_{G'}(u)-b(u)$ internal vertices. We construct part $E_{u_i}$ of $\gm$ from the edge set of the substitute of $u_i$ and all edges of positive weight (with regard to $\wm$) which are incident to some external vertex of $u_i$. The representation of $E_{u_i}$ consists of up to $b(u)$ copies of half-edges incident to $u$ of the greatest positive weights and one incident edge of the substitute of $u_i$ per every internal vertex of $u_i$. If there are less than $b(u)$ such half-edges, the copies of all of them are added to the representation of $E_{u_i}$. Of course, we choose this representation in such a way that its two different edges are not incident.
\item
Consider any copy $z_i$ in $G^*$ of a global vertex $z$ of $G'$. Recall that all edges incident to $z$ in $G'$ have weight zero. We construct part $E_{z_i}$ from the edge set of the substitite of $z_i$. Its representation consists of one incident egde per every internal vertex of $z_i$.
\item
Consider any vertex $v$ of $G$. For every copy $v_i$ of $v$ in $G^*$, we construct part $E_{v_i}$ from the edge set of the substitute of $v_i$. Moreover, for every added path $P$ of $v$, we construct part $E_P$ from all edges of $P$ and all edges of the substitutes of both central vertices of $P$. Recall that $v$ is incident to at most one half-edge in $G'$. Hence, at most one proper external vertex $u_0$ of $v_i$ is incident to some edge which belongs to the representation of some part of $\gm$. Observe that, if $u_0$ exists, we cannot include any edge incident to $u_0$ in the representation of $E_{v_i}$. The construction of the representations of these parts depends on the degree of $v$ in $G$.
\begin{itemize}
\item
Consider a case where $\deg_G(v)=t+1$. Notice that the substitute of any copy $v_i$ of $v$ consists of $t$ internal vertices, $t+1$ proper external vertices and $t$ improper external vertices. The representation of $E_{v_i}$ consists of exactly one edge incident to $u$ and some proper external vertex of $u_i$ different than $u_0$, per every internal vertex $u$ of $v_i$.
\item
Consider a case where $d=\deg_{G}(v) \leq t$. Notice that the substitute of any copy $v_i$ of $v$ consists of $d$ internal vertices, $d$ proper external vertices and $d$ improper external vertices. If $u_0$ does not exist, the representation of $E_{v_i}$ consists of one edge incident to $u$ and some proper external vertex of $v_i$, per every internal vertex $u$ of $v_i$. If $u_0$ exists, the representation of $E_{v_i}$ consists of $d$ edges. Every edge of this representation is incident to some internal vertex of $v_i$. On the other hand, $d-1$ of these edges are incident to $d-1$ proper external vertices of $v_i$ different than $u_0$. One remaining edge is incident to an improper external vertex $v_i'$ of $v_i$. We choose $v_1'$ and $v_2'$ in such a way that both of them are incident to the outermost edges of the same added path of $v$.
\end{itemize}
For every added path $P$ of $v$, we choose the representation of $E_P$ in the following way. If both outermost edges of $P$ are incident to $v_1'$ or $v_2'$, the representation of $E_P$ consists of the middle edge of $P$ and one edge from the edge set of every substitute of a center of $P$. Otherwise, the representation of $E_P$ consists of the outermost edges of $P$ and, again, one edge from the edge set of every substitute of a center of $P$. Notice that the total weight of the representation of $E_P$ is equal to $6W$.
\item
Consider any edge $e$ of $\gm$ which does not belong to any part of $\gm$ yet. Observe that $\wm(e)\leq 0$. We construct part $E_e$ which consists of single edge $e$. Its representation is empty.
\end{enumerate}

It remains to show that the number of vertices of $\gm$ unmatched in $M_1$ is $\mathcal{O}(n)$. Notice that all internal and improper external vertices in $\gm$ are matched in $M_1$. Hence, it is sufficient to consider only proper external vertices. Observe that, for every copy $u_i$ in $G^*$ of any subdivision vertex $u$ of $G'$, at most $b(u)$ proper external vertices of $u_i$ are unmatched in $M_1$. Moreover, for every copy $z_i$ in $G^*$ of any global vertex $z$ of $G'$, there are exactly $b(z)$ proper external vertices of $z_i$ unmatched in $M_1$. Therefore, there are $\mathcal{O}(t)$ proper external vertices unmatched in $M_1$ per every gadget in $G'$. For any copy $v_i$ in $G^*$ of any vertex $v$ of $G$, at most one proper external vertex of $v_i$ is unmatched in $M_1$. It exists only if $\deg_G(v)=t+1$ and no proper external vertex of $v_i$ is incident to a half-edge which belongs to the representation \mwcom{maybe add some figures to this subsection?} of some part of $\gm$. Therefore, there are $\mathcal{O}(n)$ proper external vertices in $\gm$ which are unmatched in $M_1$.

\subsection{Computing a minimum weight \lbmatching{} of $G'$ in the unweighted case}\label{subsec:running-unweighted}
In this subsection, we present how to find a minimum weight \lbmatching{} of $G'$ in $\mathcal{O}(\sqrt{n}m)$ time in the unweighted case, i.e. when $w(e)=1$ for every edge $e$ of $G$. Notice that $w$ is vertex-induced in this case since we can set the potential of every vertex of a forbidden subgraph to $\frac{1}{2}$. Hence, the weight $w'(e)$ of every half-edge $e$ of $G'$ is defined as $\frac{1}{2}$. We claim that the minimum weight \lbmatching{} of $G'$ is a {\bf \em minimum cardinality} \lbmatching{} of $G'$, i.e. an \lbmatching{} of $G'$ with the minimum number of edges among all \lbmatching{}s of $G'$. If follows from the following lemma.

\begin{lemma}
Let $M'$ be any \lbmatching{} of $G'$. Then
$$|M'|-w'(M')=k_{t+1}+k_{t,t}+(p-1)k^p_q+\sum_{\textrm{$H$ -- a dense subgraph of $G$}}(p-\frac{|C_H|}{2}+1),$$
where $k_{t+1}$, $k_{t,t}$ and $k^p_q$ denote the number of the gadgets for, correspondingly, \kt{}'s, \ktt{}'s and \kpq{}'s in $G'$.
\end{lemma}
\begin{proof}
Notice that every edge of $M'$ which does not belong to any gadget contributes one to both $|M'|$ and $w'(M')$. Hence, $|M'|-w'(M')$ is equal to the sum of values $|M_H'|-w'(M_H')$ for all forbidden subgraphs $H$ of $G$ for which the gadget was added in $G'$, where $M_H'$ denotes a set of all edges of the gadget for $H$ which belong to $M'$. Notice that for every problematic \kt{} and \ktt{} $H$ of $G$, $M_H'$ consists of exactly two half-edges, hence $|M_H'|-w'(M_H')=1$ then. Similarly, for every $H=\kpq{}$ of $G$, $|M_H'|=p$ and $w'(M_H')=1$. In the end, for any dense subgraph $H$ of $G$ with the core of size $2k$, we have $|M_H'|=p-k+2$ and $w'(M_H')=1$.
\end{proof}

We define a {\bf \em maximum cardinality} \lbmatching{} of $G'$ analogously. To calculate a minimum cardinality \lbmatching{} of $G'$ we use the following algorithm given by Gabow.

\begin{theorem}[\cite{Gabow1983}]\label{thm:gabow-unweighted}
There is an algorithm that, given a multigraph~$G=(V,E)$ (i.e. $G$ may contain parallel edges) and vectors $l,b\in\mathbb{N}^V$, finds a maximum cardinality \lbmatching{} of $G$ in  time $O\left(\sqrt{\sum_{v\in V}b(v)}|E|\right)$, assuming that $G$ admits any \lbmatching{}.
\end{theorem}

It turns out that calculating a minimum cardinality \lbmatching{} can easily be reduced to calculating a maximum cardinality \lbmatching{}.

\begin{lemma}\label{lem:running-minimum-cardinality}
There is an algorithm that, given a multigraph~$G=(V,E)$ (i.e. $G$ may contain parallel edges) and vectors $l,b\in\mathbb{N}^V$, finds a {\bf \em minimum} cardinality \lbmatching{} of $G$ in  time $O\left(\sqrt{\sum_{v\in V}b(v)}|E|\right)$, assuming that $G$ admits any \lbmatching{}.
\end{lemma}
\begin{proof}
We construct an auxiliary graph $G^*$ together with vectors $l^*,b^*\in\mathbb{N}^{V(G^*)}$. We start off with $G^*=G$. For every vertex $v$ of $G$, we add an auxiliary vertex $v'$ to $G^*$. We connect $v$ with $v'$ by $b(v)-l(v)$ parallel edges in $G^*$. We set $l^*(v)=b^*(v)=b(v)$, $l^*(v')=0$ and $b^*(v')=b(v)-l(v)$. We claim that every \lbmatching{} of $G$ corresponds to an $(l^*,b^*)$-matching of $G^*$. Indeed, for any \lbmatching{} $M$ of $G$, we can construct an $(l^*,b^*)$-matching $M^*$ of $G^*$ by adding, for every vertex $v$ of $G$, $b(v)-\deg_M(v)$ copies of an edge $(v,v')$ to $M$. Notice that
$$ |M^*| = |M| + \sum_{v\in V}(b(v)-\deg_M(v)) = \left(\sum_{v\in V}b(v)\right) - |M|, $$
so a minimum cardinality \lbmatching{} of $G$ corresponds to a maximum cardinality $(l^*,b^*)$-matching of $G^*$. Hence, we can calculate the solution using the algorithm from Theorem~\ref{thm:gabow-unweighted} on $G^*$, $l^*$ and $b^*$. Notice that $\sum_{v\in V(G*)}b^*(v) \leq 2\sum_{v\in V(G)}b(v)$. Moreover, we can assume that $b(v)\leq\deg_{G}(v)$ for every vertex $v$ of $G$. Both of these facts guarantee the desired running time.
\end{proof}

We can calculate a minimum cardinality \lbmatching{} of $G'$ by applying Lemma~\ref{lem:running-minimum-cardinality} directly. However, it works in $\mathcal{O}(m^{3/2})$ time since
the sum of the upper bounds of the capacity intervals in $G'$ is $O(m)$. We can calculate it faster using the following lemma.

\begin{lemma}\label{lem:running-bounded-degree}
Let $M'$ be any \lbmatching{} of $G'$. Then there exists a minimum cardinality \lbmatching{} $N'$ of $G'$ such that $\deg_{N'}(v) \leq \deg_{M'}(v)$ for every vertex $v$ of $G'$.
\end{lemma}

We postpone the proof of Lemma~\ref{lem:running-bounded-degree} since it is quite technical. The main idea behind calculating a minimum cardinality \lbmatching{} of $G'$ is to find some \lbmatching{} $M'$ of $G'$ which consists of $\mathcal{O}(n)$ edges. Such $M'$ can be constructed in the following way. We initialize $M'$ as the empty set. For every vertex $v$ of degree $t+1$ in $G$, we add any edge of $G$ incident to $v$ to $M'$. Additionally, we have to add some edges of the gadgets in order to satisfy the degree constrains of the global and subdivision vertices. It is easy to check that $O(t)$ edges of every gadget are sufficient, hence $M'$ has $\mathcal{O}(n)$ edges in total. Using Lemma~\ref{lem:running-bounded-degree}, we can find a minimum cardinality \lbmatching{} of $G'$ by running the algorithm from Lemma~\ref{lem:running-minimum-cardinality} for $G$ and vectors $l,b'\in\mathbb{N}^{V(G)}$ such that $b'(v)=\deg_M(v)$ for every vertex $v$ of $G'$. Since $\sum_{v\in V}{b'(v)}=2|M'|=\mathcal{O}(n)$, the total runtime of this procedure is $\mathcal{O}(\sqrt{n}m)$.

To prove Lemma~\ref{lem:running-bounded-degree}, we use the classical notion of alternating paths and cycles. Let $M$ be an \lbmatching{} of a graph $G$. An edge belonging to $M$ will be referred to as an {\bf \em $M$-edge} and an edge not belonging to $M$ as a {\bf \em non-$M$-edge}. {\bf \em An $M$-alternating path} $P$ is any sequence of vertices $(v_1,v_2,\ldots,v_k)$ such that edges on $P$ are alternately $M$-edges and non-$M$-edges and no edge occurs on $P$ more than once.
{\bf \em An $M$-alternating cycle} $C$ has the same definition as an $M$-alternating path except that $v_1=v_k$ and additionally $(v_{k-1},v_k)\in M$ if and only if $(v_1,v_2)\notin M$. Note that an $M$-alternating path or cycle may go through some vertices more than once but via different edges.
An $M$-alternating path is called {\bf \em $M$-augmenting} if it begins and ends with a non-$M$-edge. With a slight abuse of notation, we identify an $M$-alternating path or cycle with its edge set. For subsets $A,B\subseteq E$, we define a {\bf \em symmetric difference} $A\oplus B=(A\setminus B)\cup(B\setminus A)$.

\begin{proof}[\textbf{\textup{Proof of Lemma~\ref{lem:running-bounded-degree}}}]
Let $N$ be any minimum cardinality \lbmatching{} of $G'$. It is well-known \mwcom{add some source} that $M'\oplus N$ can be partitioned into a collection $\mathcal{P}$ of pairwise edge-disjoint paths and cycles which are both $M'$-alternating and $N$-alternating. We assume that partition $\mathcal{P}$ is {\bf \em minimal}, i.e. no two paths or cycles of $\mathcal{P}$ can be merged into one $M'$-alternating and $N$-alternating path or cycle. From the definition of $N$, $|N|\leq |M'|$, so $\mathcal{P}$ contains $k=|M'|-|N|$ different paths $P_1$, $P_2$, \ldots, $P_k$ which are both $M'$-alternating and $N$-augmenting. We say that these paths are {\bf \em special}. Define $N'=M'\oplus P_1\oplus P_2 \oplus\ldots\oplus P_k$, i.e. $N'$ consists of all $M'$-edges which do not belong to any special path and all $N$-edges which belong to some special path. We claim that $N'$ is the sought-after minimum cardinality \lbmatching{} of $G'$. At first, notice that $|N'|=|M'|-k=|N|$. Moreover, observe that adding $P_i$ to the symmetric difference with $M'$ decrease only the degree of both endpoints of $P_i$ while leaving the degrees of other vertices of $G'$ unchanged. Therefore, $\deg_{N'}(v)\leq\deg_{M'}(v)\leq b(v)$ for every vertex $v$ of $G'$.

It remains to show that $\deg_{N'}(v)\geq l(v)$ for every vertex $v$ of $G'$. From our previous observation, it is sufficient to prove this only for endpoints of special paths. Consider any vertex $v$ of $G'$ which is an endpoint of $p\geq 1$ special paths. If some special path has both endpoints in $v$, we calculate it twice. Notice that $v$ is incident to exactly $p$ different $M'$-edges in $G'$ which belong to some special paths. We claim that the number of remaining $N$-edges incident to $v$ in $G'$ is not greater than the number of remaining $M'$-edges incident to $v$ in $G'$. Indeed, if there are more $N$-edges than $M'$-edges among the remaining edges incident to $v$ in $G'$, then $v$ is an endpoint of some $M'$-alternating and $N$-alternating path $P'$ of $\mathcal{P}$ such that an edge of $P'$ which is incident to $v$ in $G'$ belongs to $N$. Notice that we can merge $P'$ and any special path which has an endpoint in $v$ into one path -- a contradition with the minimality of $\mathcal{P}$. From our claim, $\deg_{N}(v)\leq \deg_{M'}(v)-p$. Hence, $\deg_{N'}(v)=\deg_{M'}(v)-p \geq \deg_{N}(v) \geq l(v)$ because $N$ is an \lbmatching{} of $G'$.
\end{proof}

\subsection{Finding all forbidden subgraphs of $G$}\label{sec:running-kpq}
In this subsection, we present Algorithm~\ref{alg:find-kpq-main} which finds all \kpq{}'s of $G$ in $\mathcal{O}(m)$ time, given some natural numbers $p\geq 2$ and $q\geq 1$ such that $t=(p-1)q$. It also returns which pairs of \kpq{}'s of $G$ have some common vertices. Applying this algorithm to Algorithm~\ref{alg:main} for the bounded \kpq{}-free $t$-matching problem is fairly straightforward. For the bounded restricted $t$-matching problem, we run Algorithm~\ref{alg:find-kpq-main} twice -- once to find all \kt{}'s of $G$, and again to find all \ktt{}'s of $G$. However, we have to find also all pairs consisting of \kt{} and \ktt{} of $G$ which have a common vertex. This can be done in $\mathcal{O}(m)$ time since Lemma~\ref{lem:cliques-kt-ktt-description} indicates that for every such pair, $t=3$ and the vertex set of \kt{} is a subset of the vertex set of \ktt{}.

We say that a graph is {\bf \em bounded} if its every vertex has degree at most $t+1$. The following lemma gives an importart subprocedure of Algorithm~\ref{alg:find-kpq-main}. We postpone its proof since it is quite technical. 

\begin{lemma}\label{lem:find-kpq-naive}
Let $v$ be any vertex of a bounded graph $\tilde{G}$. We can find in $\mathcal{O}(t^2)$ time all \kpq{}'s of $\tilde{G}$ which contain $v$.
\end{lemma}

The main idea behind Algorithm~\ref{alg:find-kpq-main} is to find two different vertices $v_1$ and $v_2$ of $G$ which have $\Theta(t)$ common neighbours using Lemma~\ref{lem:find-kpq-friend} given below. We check if $v_1$ or $v_2$ belongs to some \kpq{} of $G$ using Lemma~\ref{lem:find-kpq-naive}. Recall that if some two \kpq{}'s of $G$ have a common vertex, then they have at least $pq-2$ common vertices (see Lemmas~\ref{lem:cliques-kt-description}, \ref{lem:cliques-ktt-description}, \ref{lem:partite-description} and~\ref{lem:partite-2-description}). Therefore, if we find some $H=K^p_q$ of $G$, we can find all \kpq{}'s of $G$ which have some common vertex with $H$ in $\mathcal{O}(t^2)$ time. In such a case, we return $H$ and all \kpq{}'s of $G$ which have a common vertex with $H$. Then we remove the vertex set of $H$ from $G$ and we continue to search for \kpq{}'s in the rest of the graph. Observe that if $v_1$ and $v_2$ do not belong to any \kpq{} of $G$, then none of the common neighbours of $v_1$ and $v_2$ belongs to any \kpq{} of $G$ as well. In this case, we remove $v_1$, $v_2$, and all their common neighbours from $G$, and we search for \kpq{}'s in the rest of the graph. In both cases, we remove $\Theta(t)$ vertices in $\mathcal{O}(t^2)$ time, which means that we spend $\mathcal{O}(t)$ time per vertex on average. Since no vertex of degree less than $t$ belongs to some \kpq{} of $G$, we can remove such vertices from $G$ at the beginning. Therefore, we can conclude that the running time of Algorithm~\ref{alg:find-kpq-main} is $\mathcal{O}(m)$.

\begin{lemma}\label{lem:find-kpq-friend}
Let $v_1$ be any vertex of a bounded graph $\tilde{G}$. There exists a procedure which works in $\mathcal{O}(t)$ time and it either concludes that $v_1$ does not belong to any \kpq{} of $\tilde{G}$, or returns a vertex $v_2\neq v_1$ of $\tilde{G}$ which has at least $\max\{p-2,1\}q$ common neighbours with $v_1$.
\end{lemma}
\begin{proof}
The procedure works in the following way. It processes any two different neighbours of $v_1$, say $u_1$ and $u_2$ (or all neighbours of $v_1$, if there are less than two of them). The procedure processes every vertex $u\in\{u_1,u_2\}$ as follows. It checks any two different neighbours of $u$ different than $v_1$, say $z_1$ and $z_2$, as candidates for $v_2$ (or all of them, if there are less than two such vertices). For every $z\in\{z_1,z_2\}$, the procedure checks if $v_1$ and $z$ have at least $\max\{p-2,1\}q$ common neighbours. If so, the procedure returns $v_2=z$. If none of the vertices $z$ satisfies this condition, the procedure returns that $v_1$ does not belong to any \kpq{} of $\tilde{G}$. It is easy to notice that this procedure can be implemented to run in $\mathcal{O}(t)$ time. Its correctness follows from the following claims.
\begin{claim}
Assume that $v_1$ belongs to some $H=K^p_q$ of $\tilde{G}$. Then the procedure given above considers some $u$ and $z$ such that both edges $(v_1,u)$ and $(u,z)$ belong to $H$.
\end{claim}
\begin{proof}
We consider the execution of the presented procedure. Since $v_1$ belongs to $t$-regular $H$, $v_1$ has at least $t\geq 3$ neighbours in $\tilde{G}$, so $u_1$ and $u_2$ exist. Recall that $\tilde{G}$ is bounded, so $\deg_{\tilde{G}}(v_1) \leq t+1$. Therefore, at most one edge of $\tilde{G}$ incident to $v_1$ does not belong to $H$. It implies that at least one of edges $(v_1,u_1)$ and $(v_1,u_2)$ belongs to $H$. Consider $u\in\{u_1,u_2\}$ such that $(v_1,u)$ belongs to $H$. We can use similar arguments for $u$ as for $v_1$ to prove that both $z_1$, $z_2$ exist and that at least one of edges $(u,z_1)$, $(u,z_2)$ belongs to $H$. We choose $z\in\{z_1,z_2\}$ such that $(u,z)$ belongs to $H$.
\end{proof}
\begin{claim}
Assume that some two different edges $(w_1,t)$ and $(w_2,t)$ belong to the same \kpq{} of $\tilde{G}$. Then $w_1$ and $w_2$ have at least $\max\{p-2,1\}q$ common neighbours.
\end{claim}
\begin{proof}
Let $H$ be any \kpq{} of $G$ which contains both $(w_1,t)$ and $(w_2,t)$.
Assume that $t\in V_1(H)$. Then both $w_1$ and $w_2$ belong to $V(H)\setminus V_1(H)$. We consider the following two cases.
\begin{enumerate}
\item If $w_1$ and $w_2$ belong to the same color class of $H$, say $V_2(H)$, then they have at least $(p-1)q$ common neighbours -- all vertices of $V(H)\setminus V_2(H)$.
\item\label{case:kpq-different} If $w_1$ and $w_2$ belong to different color classes of $H$, say $V_2(H)$ and $V_3(H)$ respectively, then they have at least $(p-2)q$ common neighbours -- all vertices of $V(H)\setminus(V_2(H)\cup V_3(H))$.
\end{enumerate}
In both cases, $w_1$ and $w_2$ have at least $(p-2)q$ common neighbours. However, case~2. cannot happen if $p=2$, so we can conclude that $w_1$ and $w_2$ have at least $\max\{p-2,1\}q$ common neighbours.
\end{proof}
\end{proof}

Now we present a pseudocode of Algorithm~\ref{alg:find-kpq-main} and a proof of its correctness.

\begin{algorithm}[H]
\begin{algorithmic}[1]
\State \label{itm:alg2-step1}Initialize $\tilde{G}$ as $G$. Remove all vertices of degree less than $t$ in $G$ from $\tilde{G}$.
\State \label{itm:alg2-step3}If $\tilde{G}$ is empty, stop the algorithm.
\State \label{itm:alg2-step5}Run the procedure from Lemma~\ref{lem:find-kpq-friend} on $\tilde{G}$ and any its vertex $v_1$.
\State \label{itm:alg2-step6}If $v_1$ does not belong to any \kpq{} of $\tilde{G}$, remove $v_1$ from $\tilde{G}$ and go to Step~\ref{itm:alg2-step3}.
\State \label{itm:alg2-step7}Otherwise, let $v_2$ be a vertex found in Step~\ref{itm:alg2-step5}. Run the procedure from Lemma~\ref{lem:find-kpq-naive} on $\tilde{G}$ and $v_i$, for every $i\in\{1,2\}$.
\State \label{itm:alg2-step8}If none of the vertices $v_1$, $v_2$ belong to any \kpq{} of $\tilde{G}$, remove $v_1$, $v_2$ and all its common neighbours from $\tilde{G}$ and go to Step~\ref{itm:alg2-step3}.
\State \label{itm:alg2-step9}Otherwise, let $H$ be any \kpq{} found in Step~\ref{itm:alg2-step7}. Return $H$ and all \kpq{}'s of $\tilde{G}$ which have some common vertex with $H$. Return that every two different returned \kpq{}'s have a common vertex. Remove $H$ from $\tilde{G}$ and go to Step~\ref{itm:alg2-step3}.
\end{algorithmic}
\caption{Finding all \kpq{}'s of $G$}
\label{alg:find-kpq-main}
\end{algorithm}

\begin{claim}
Algorithm~\ref{alg:find-kpq-main} finds all \kpq{}'s of $G$ in $\mathcal{O}(m)$ time. Moreover, it finds all pairs of \kpq{}'s of $G$ which have a common vertex.
\end{claim}
\begin{proof}
At first, we prove that Algorithm~\ref{alg:find-kpq-main} finds all \kpq{}'s of $G$. We claim that during its execution an invariant is maintained that all \kpq{}'s of $G$ which were not returned yet are present in $\tilde{G}$. To show this, it is sufficient to prove that every time some vertex $v$ is removed from $\tilde{G}$, all \kpq{}'s of $\tilde{G}$ which contain $v$ were returned by Algorithm~\ref{alg:find-kpq-main} already. We consider all steps of Algorithm~\ref{alg:find-kpq-main} where some vertices are removed from $\tilde{G}$:
\begin{itemize}
\item Step~\ref{itm:alg2-step1}: notice that no vertex removed from $\tilde{G}$ belongs to any \kpq{} of $G$ since \kpq{} is $t$-regular,
\item Step~\ref{itm:alg2-step6}: we remove $v_1$ which does not belong to any \kpq{} of $\tilde{G}$,
\item Step~\ref{itm:alg2-step8}: if some common neighbour of $v_1$ and $v_2$ belongs to some \kpq{} of $\tilde{G}$, then also $v_1$ or $v_2$ belongs to some \kpq{} of $\tilde{G}$, but we assume that both $v_1$ and $v_2$ do not belong to any \kpq{} of $\tilde{G}$,
\item Step~\ref{itm:alg2-step9}: we remove $V(H)$ from $\tilde{G}$, but we return all \kpq{}'s of $\tilde{G}$ which contain some vertex of $V(H)$. 
\end{itemize}

Now we prove that Algorithm~\ref{alg:find-kpq-main} finds all pairs of \kpq{}'s of $G$ which have a common vertex. We return some \kpq{}'s only in Step~\ref{itm:alg2-step9}. Recall that all these \kpq{}'s do not belong to $\tilde{G}$ after Step~\ref{itm:alg2-step9} anymore. Notice that it is sufficient to show that all these \kpq{}'s do not have any common vertices with all other \kpq{}'s of $\tilde{G}$. It follows from the following claim.
\begin{claim}
Let $H_1$, $H_2$ and $H_3$ be any \kpq{}'s of a bounded graph $\tilde{G}$. Assume that $H_1$ has a common vertex with $H_2$ and a common vertex with $H_3$. Then $H_2$ and $H_3$ have a common vertex.
\end{claim}
\begin{proof}
Recall that, by Lemmas~\ref{lem:cliques-kt-description}, \ref{lem:cliques-ktt-description}, \ref{lem:partite-description} and~\ref{lem:partite-2-description}, any two different \kpq{}'s of $\tilde{G}$ which have a common vertex have at least $pq-2$ common vertices. Moreover, if $p > 2$, then they have at least $pq-1$ common vertices.

Consider a case where $p > 2$. Then at most one vertex of $H_1$ does not belong to $H_2$ and at most one vertex of $H_1$ does not belong to $H_3$. Hence, there are at least $pq-2=t+q-2\geq 2$ vertices of $H_1$ which belong to both $H_2$ and $H_3$. Any of these vertices is a common vertex of $H_2$ and $H_3$.

Consider a case where $p = 2$. Then at most two vertices of $H_1$ do not belong to $H_2$ and at most two do not belong to $H_3$. Hence, there are at least $2q-4=2t-4 \geq 2$ vertices of $H_1$ which belong to both $H_2$ and $H_3$.
\end{proof}

It remains to show that Algorithm~\ref{alg:find-kpq-main} works in $\mathcal{O}(m)$ time. It is easy to implement Step~\ref{itm:alg2-step1} in $\mathcal{O}(m)$ time. We divide the rest of the execution of Algorithm~\ref{alg:find-kpq-main} into phases. Every phase works in $\mathcal{O}(tk)$ time and ends with the removal of $k$ vertices from $\tilde{G}$, for some natural positive number $k$. Notice that it implies that Algorithm~\ref{alg:find-kpq-main} works in $\mathcal{O}(tn_t)$ time where $n_t$ denotes the number of the vertices of degree at least $t$ in $G$, but from the handshaking lemma, $tn_t\leq 2m$.

The partition of the rest of the execution of Algorithm~\ref{alg:find-kpq-main} into phases is simple. For every removal of some vertices from $\tilde{G}$, its phase consists of this removal and all steps preceding it until the previous removal. We prove that every such phase works in $\mathcal{O}(tk)$ time indeed, where $k$ denotes the number of removed vertices. For the removal in Step~\ref{itm:alg2-step6}, by Lemma~\ref{lem:find-kpq-friend}, its phase works in $\mathcal{O}(t)$ time and we remove one vertex. For the removal in Steps~\ref{itm:alg2-step8} and~\ref{itm:alg2-step9}, by Lemmas~\ref{lem:find-kpq-naive},~\ref{lem:find-kpq-friend} and our previous observations, its phase works in $\mathcal{O}(t^2)$ time and we remove $\Theta(t)$ vertices.

\end{proof}

Now we present the proof of Lemma~\ref{lem:find-kpq-naive}. We prove it for cases where $\kpq=\kt$, $\kpq=\kpt$ and for the remaining cases separately. At first, we give some definitions. Let $G=(V,E)$ be any graph. A {\bf \em complement} of $G$ is a graph $H=(V,E')$, where $E'$ consists of all pairs of different vertices of $V$ which are not connected by any edge from $E$. Graph $G$ is said to a {\bf \em star} if $E=\{(v_0,v):v\in V\setminus\{v_0\}\}$ for some vertex $v_0$ of $V$. In such a case, vertex $v_0$ is called a {\bf \em center} of $G$.
Graph $G$ is said to be {\bf \em empty} if its edge set is empty. For vertex-disjoint sets $A,B\subseteq V$, $G[A,B]$ denotes a bipartite graph $(A\cup B,E')$ where $E'$ consists of all edges of $G$ whose one endpoint belongs to $A$ and another one to $B$. A {\bf \em bipartite complement} of $G[A,B]$ is a bipartite graph $(A\cup B,E'')$ where $E''=\{(a,b):a\in A\land b\in B\land (a,b)\notin E\}$. Any vertex $v\in V$ is said to be a {\bf \em cut vertex} of $G$ if a removal of $v$ from $G$ increases the number of the connected components of $G$. A cut vertex $v$ is said to be {\bf \em nontrivial} if every connected component which appears after the removal of $v$ from $G$ consists of at least two vertices.
We extend the definition of \kpq{} to cases where $p\in\{0,1\}$ naturally, i.e. $K^0_q$ is a graph whose vertex set is empty whereas $K^1_q$ is an empty graph which consists of $q$ vertices.

\begin{proof}[\textbf{\textup{Proof of Lemma~\ref{lem:find-kpq-naive} for \kt{}'s}}]
Let $e_1$ and $e_2$ be any two different edges incident to $v$ in $\tilde{G}$. (If $v$ has the degree less than two in $\tilde{G}$, then $v$ clearly cannot belong to any $t$-regular \kt{} of $\tilde{G}$ since $t\geq 3$.) Because the degree of $v$ in $\tilde{G}$ is not greater than $t+1$, any \kt{} of $\tilde{G}$ which contains $v$ has to contain $e_1$ or $e_2$. For each edge $e\in\{e_1,e_2\}$, we find all \kt{}'s of $\tilde{G}$ which contain $e$. Observe that if $e$ belongs to some \kt{} of $\tilde{G}$, then both its endpoints must have at least $t-1$ common neighbours.

If both endpoints of $e$ have exactly $t-1$ common neighbours, we check if these common neighbours and both endpoints of $e$ form some \kt{} of $\tilde{G}$. Observe that checking if a given subset $A$ of $t+1$ vertices form a \kt{} in a bounded graph $\tilde{G}$ can be implemented to run in $\mathcal{O}(t^2)$ time -- we can construct the complement of $\tilde{G}[A]$ and check if it contains no edges. Indeed, for every vertex $u$ of $A$, we can mark all its neighbours in $\tilde{G}$ in $\mathcal{O}(t)$ time. All unmarked vertices of $A\setminus\{u\}$ are all the neighbours of $u$ in the complement of $\tilde{G}[A]$.

If both endpoints of $e$ have exactly $t$ common vertices, we check if they form a $K_{t+2}$ of $\tilde{G}$ together with both endpoints of $e$. If not, we make use of the following observation to find all \kt{}'s in a subgraph of $\tilde{G}$ induced by these vertices.
\begin{observation}\label{obs:cliques-kt-stars}
Let $A$ be any set of $t+2$ different vertices of $\tilde{G}$. Let $H$ be any \kt{} of $\tilde{G}$ such that $V(H)\subseteq A$. Then the complement of $\tilde{G}[A]$ is a star with a center in the only vertex of $A$ which does not belong to $H$.
\end{observation}
Therefore, it is sufficient to check if the complement of $\tilde{G}$ induced on these $t+2$ vertices is a star. If so, we get at most two different \kt{}'s of $\tilde{G}$ which contain $e$, depending on how many centers this star has. (Notice that a star consisting of a single edge has two different centers.) Again, we can construct such complement in $\mathcal{O}(t^2)$ time.
\end{proof}

\begin{proof}[\textbf{\textup{Proof of Lemma~\ref{lem:find-kpq-naive} for \kpq{}'s} for $p\geq 2$ and $q\geq 3$}]
We consider any two different edges $e_1$ and $e_2$ incident to $v$ in $\tilde{G}$, similarly as in the proof for \kt{}'s. We find all \kpq{}'s of $\tilde{G}$ which contain $e_1$ or $e_2$.

\begin{definition}
Let $e=(v_0,v_1)$ be any edge of $\tilde{G}$. We denote by $\Gamma(e)$ a set of common neighbours of both endpoints of $e$. For $i\in\{0,1\}$, we denote by $\Gamma_i(e)$ a set of all neighbours of $v_i$ different than $v_{1-i}$ which do not belong to $\Gamma(e)$.
\end{definition}

Observe that the endpoints of any edge $e$ of $\tilde{G}$ belong to neither $\Gamma(e)$ nor $\Gamma_i(e)$. Since $\tilde{G}$ is bounded, $|\Gamma(e)|+|\Gamma_i(e)| \leq t$ for every $i\in\{0,1\}$.

\begin{observation}\label{obs:partite-common-part}
For any edge $e$ of any \kpq{} of $\tilde{G}$, $(p-2)q\leq|\Gamma(e)|\leq (p-2)q+2$.
\end{observation}

Because of Observation~\ref{obs:partite-common-part}, it is sufficient to consider the following three cases while considering an edge $e=(v_0,v_1)\in\{e_1,e_2\}$ as a candidate for an edge of some \kpq{} of $\tilde{G}$. We assume that both $v_0$ and $v_1$ have degree exactly $t+1$ in $\tilde{G}$ since we can add artificial vertices if necessary.

\begin{enumerate}
\item\label{itm:running1} $|\Gamma(e)|=(p-2)q+2$ and $|\Gamma_i(e)|=q-2$.

Since the number of all neighbours of each endpoint of $e$ (including $v_0$ and $v_1$) is not greater than $pq$, we can find sought-after \kpq{} using the following observation.
\begin{observation}\label{obs:partite-3-search}
Let $A$ be any set of $pq$ different vertices of $\tilde{G}$. Let $H$ be any \kpq{} of $\tilde{G}$ such that $V(H)\subseteq A$. Then the complement of $G[A]$ consists of exactly $p$ connected components, each of which consists of exactly $q$ vertices. Moreover, the vertex set of each of these components is some color class of $H$.
\end{observation}
\item $|\Gamma(e)|=(p-2)q$ and $|\Gamma_i(e)|=q$.

We check if there exists a $K^{p-2}_q$ $H'$ in $\tilde{G}[\Gamma(e)]$ similarly as in Case~\ref{itm:running1}. If so, we check if it is possible to extend $H'$ to a valid \kpq{}, using $q$ vertices of each set $\Gamma_i(e)\cup\{v_{1-i}\}$ as another color class. To do this, we search for \ktt{}'s of $\tilde{G}$ whose each color class is contained in $\Gamma_i(e)\cup\{v_{1-i}\}$, for some $i\in\{1,2\}$, using the following observation.
\begin{observation}
Let $A_1$ and $A_2$ be any two vertex-disjoint subsets of vertices of $\tilde{G}$, each of size exactly $t+1$. Let $H$ be any \ktt{} of $\tilde{G}$ such that $V_i(H)\subseteq A_i$,  for every $i\in\{1,2\}$. Then the bipartite complement of $\tilde{G}[A_1,A_2]$ consists of two (possibly empty) stars with centers in vertices of, respectively, $A_1\setminus V_1(H)$ and $A_2\setminus V_2(H)$.
\end{observation}
We check if it is possible to merge $H'$ and found \ktt{}'s into a single \kpq{} of $\tilde{G}$.
Since each set $\Gamma_i(e)\cup\{v_{1-i}\}$ contains $q+1$ vertices, we can find several valid \kpq{}'s of $\tilde{G}$ in this case. \\
\item $|\Gamma(e)|=(p-2)q+1$ and $|\Gamma_i(e)|=q-1$.

We check if there exists a $K^{p-2}_q$ $H'$ in $\tilde{G}[\Gamma(e)]$ using the following observation.
\begin{observation}\label{obs:partite-3-search-plus}
Let $A$ be any set of $pq+1$ different vertices of $\tilde{G}$. Let $H$ be any \kpq{} of $\tilde{G}$ such that $V(H)\subseteq A$. Then the only vertex of $A$ which does not belong to $H$ is only either isolated vertex or nontrivial cut vertex of the complement of $\tilde{G}[A]$.
\end{observation}
We check if both sets $\Gamma_i(e)\cup\{v_{1-i}\}$ form other color classes of sought-after \kpq{}. Moreover, we check if it possible to replace some vertex of one of these sets by the only vertex of $\Gamma(e)$ which does not belong to $H'$ to get a valid \kpq{} of $\tilde{G}$.
\end{enumerate}
\end{proof}

\begin{proof}[\textbf{\textup{Proof of Lemma~\ref{lem:find-kpq-naive} for \kpt{}'s} for $p\geq 3$}]
We find all \kpt{}'s of $\tilde{G}$ which contain $v$ similarly to \kpq{}'s. We need the following slight modifications to Observation~\ref{obs:partite-3-search} and Observation~\ref{obs:partite-3-search-plus}.

\begin{observation}
Let $A$ be any subset of $2p$ different vertices of $\tilde{G}$. Let $H$ be any \kpt{} of $\tilde{G}$ such that $V(H)\subseteq A$. Then the complement of $\tilde{G}[A]$ is a (possibly empty) matching.
\end{observation}

\begin{observation}
Let $A$ be any set of $2p+1$ different vertices of $\tilde{G}$. Let $H$ be any \kpt{} of $\tilde{G}$ such that $V(H)\subseteq A$. Then there exist a constant number of vertices $v$ such that $\tilde{G}[A\setminus\{v\}]$ is a matching. Moreover, all of these vertices can be found in $\mathcal{O}(t^2)$ time.
\end{observation}

\end{proof}

\subsection{Finding the potential functions}\label{sec:running-potential}

In this subsection, we present how to find the potential function of any forbidden subgraph of $G$.
We use the following observations to do so.
\begin{observation}\label{obs:triangle-potentials}
Consider any triangle $t$ of $G$. Then there exists exactly one potential function of $t$.
\end{observation}
The potential function of any \kt{} of $G$ can be calculated in $\mathcal{O}(t^2)$ time by covering the set of its vertices by triangles and by using Observation~\ref{obs:triangle-potentials}.  
\begin{observation}
Consider any $H=\ktt{}$ of $G$ and let $r_H$ be any potential function of $H$. For any real number $\delta$, define the following potential function of $H$:
\[ r_H^\delta(v)=\left\{\begin{array}{ll}
r_H(v)+\delta & \textrm{if $v\in V_1(H)$,}\\
r_H(v)-\delta & \textrm{if $v\in V_2(H)$.}
\end{array}\right. \]
Then the set of all potential functions of $H$ is a set of all $r_H^\delta$ over all real numbers $\delta$.
\end{observation}
We can find any potential function of any $H=\ktt{}$ of $G$ by defining the potential of some vertex of $H$ arbitrarily which determines potentials of all other vertices of $H$. It is easy to see that this procedure can be implemented to run in $\Theta(t)$ time. 

For $p\geq 3$, we can find the potential function of any \kpq{} of $G$ in $\mathcal{O}(t^2)$ time using Observation~\ref{obs:triangle-potentials} in the same way as for \kt{}'s.

\section{Conclusion}
We  presented the first polynomial algorithms for the weighted versions of the bounded  restricted and \kpq{}-free $t$-matching problems in which the weight function is vertex-induced on every forbidden subgraph. Moreover, these algorithms for the unweighted versions are faster than those known previously.

It is imaginable that our algorithms can be adapted to find  maximum size or weight $t$-matchings which contain  no complete bipartite subgraphs consisting of $t+2$ vertices in the graphs of maximum degree at most $t+1$, possibly leading to a faster algorithm for the problem of increasing connectivity by one. 

Based on our result, one can suspect that the degree sequences of maximum weight restricted $t$-matchings form an $M$-concave function on the constant-parity jump system in any graph of  maximum degree at most $t+1$ if the weight function is vertex-induced on every forbidden subgraph.

\bibliographystyle{abbrv}
\bibliography{bib}


\end{document}